\documentclass[review,authoryear]{elsarticle}
\usepackage{geometry}
\makeatletter
\def\ps@pprintTitle{%
    \let\@oddhead\@empty
    \let\@evenhead\@empty
    \def\@oddfoot{\footnotesize\itshape
         {\phantom{x}} \hfill\today}%
    \let\@evenfoot\@oddfoot
}
\makeatother

\usepackage[utf8]{inputenc}
\usepackage{parskip}
\usepackage{graphicx, color, url}
\usepackage[colorlinks=true,linkcolor=black, citecolor=blue, urlcolor=blue]{hyperref}
\usepackage{booktabs,threeparttable}
\usepackage{graphicx}
\usepackage{caption}
\usepackage{subcaption}
\usepackage{hyperref}
\graphicspath{ {./images/} }
\usepackage{diagbox}
\usepackage{amsmath}
\usepackage{amsthm}
\usepackage{setspace, geometry}
\usepackage{amssymb}
\usepackage{tabularx}
\usepackage{scrextend}
\usepackage{floatrow}

\newtheorem{theorem}{Theorem}
\newtheorem{proposition}[theorem]{Proposition}

\DeclareMathOperator*{\argmax}{argmax} 

\begin{document}

\begin{frontmatter}
\title{Computational Performance of Deep Reinforcement Learning to find Nash Equilibria\tnoteref{mytitlenote}}
\tnotetext[mytitlenote]{Christoph Graf and Claude Kl\"ockl gratefully acknowledge financial support from the Anniversary Fund of the Oesterreichische Nationalbank (OeNB), 18306. Furthermore, Christoph Graf acknowledges financial support from the Austrian Science Fund (FWF), J-3917. Johannes Schmidt and Claude Kl\"ockl   thank  the European  Research  Council  (‘reFUEL’  ERC-2017-STG  758149) for their financial support.}

\author[1]{Christoph Graf}
\ead{cgraf@stanford.edu}

\author[2]{Viktor Zobernig}
\ead{viktor.zobernig@gmail.com}

\author[2]{Johannes Schmidt}
\ead{johannes.schmidt@boku.ac.at}

\author[2]{Claude Kl\"ockl\corref{cor1}}
\cortext[cor1]{Corresponding author}
\ead{claude.kloeckl@boku.ac.at}

\address[1]{Program on Energy and Sustainable Development (PESD), Stanford University, Stanford, CA 94305-6072, United States}

\address[2]{Institute for Sustainable Economic Development, University of Natural Resources and Life Sciences, Vienna, Austria}

\begin{abstract}
We test the performance of deep deterministic policy gradient (DDPG)---a deep reinforcement learning algorithm, able to handle continuous state and action spaces---to learn Nash equilibria in a setting where firms compete in prices. These algorithms are typically considered ``model-free'' because they do not require transition probability functions (as in e.g., Markov games) or predefined functional forms. Despite being ``model-free'', a large set of parameters are utilized in various steps of the algorithm. These are e.g., learning rates, memory buffers, state-space dimensioning, normalizations, or noise decay rates and the purpose of this work is to systematically test the effect of these parameter configurations on convergence to the analytically derived Bertrand equilibrium. We find parameter choices that can reach convergence rates of up to 99\%. %with the caveat that this choice is not known a priori. %Therefore, we provide a detailed analysis on the sensitivity of parametrization on convergence and competitiveness.
The reliable convergence may make the method a useful tool to study strategic behavior of firms even in more complex settings.
\end{abstract}
\begin{keyword}
Bertrand Equilibrium \sep  Competition in Uniform Price Auctions \sep Deep Deterministic Policy Gradient Algorithm \sep Parameter Sensitivity Analysis
\end{keyword}
\end{frontmatter}

\section{Introduction}

A fundamental challenge in many applications in economics and social sciences is to derive counter-factual outcomes of the world---a necessity to, e.g., analyze the impact of a policy intervention. For many applications, this task is challenging because of the lack of measurements on certain important variables, potential measurement errors, or because the mechanisms that lead to outcomes are complex and poorly understood. If the exact mechanism of how the world operates is unknown, experiments or causal statistical models both in combination with randomization may be successfully deployed to derive valid counter-factual outcomes. There exists, however, a subset of problems in which the mechanisms are precisely defined, e.g., auctions. Although, the auction platform is typically set up centrally and operates based on a pre-determined set of rules, participants engage with each other ``through'' the auction-mechanism. The auction-clearing for given offers and bids from participants, is typically trivial to solve---a fact that is usually not true for the derivation of market participants' strategies.\footnote{Note that the unconstrained static Bertrand equilibirum where players compete in prices can also be formulated in a two-stage setting where player submit prices to an auction that will give the award (served demand) to the player with the lowest offer or split the demand in case of a tie. Although, the two-stage setting seems unnecessarily complicated for these kind of problems it will become helpful when we transition to more complex strategic interactions between players, e.g., if they are capacity constrained.}

If large comprehensive data-sets were available, counter-factual outcomes of auctions could be evaluated using non-parametric structural estimation techniques \citep[see, e.g.,][]{GuPe00, Kast11, Caou21, Char21}. Alternatively, game theoretic models could be deployed to derive counter-factual outcomes. However, in practice these models may be computationally challenging to solve because of their two-stage nature. More precisely, firms competing in an auction rely on the outcome of the second stage, i.e., the auction-clearing.\footnote{In some cases, for example, if the only ``system-condition'' is to ensure that supply equals demand, these problems can be simplified by integrating the equilibrium condition into the first stage. However, for more complicated auction-clearing mechanisms, e.g., locational pricing markets \citep[see e.g.,][]{YaAd08, GrWo20} that is no longer the case.} A relatively new approach, that we put to the test in this paper is to use reinforcement learning to model auction outcomes. 

In this framework, the environment (auction-clearing) is hard-coded consistent with its actual mode of functioning. Market participants or players interact through the environment with each other using a combination of experimenting and learning. One advantage is that this approach does not require any data; it only requires the knowledge of how the auction is organized (market rules) and a description of the agents' objective functions as well as their (physical) capacity limits. The ``data'' is then constructed on-the-fly as a combination of the strategies generated by the agents and the corresponding auction-clearing results. On first sight, this procedure parallels agent-based models which are sometimes criticized for generating any result depending only on how the researchers designed their agents. Indeed, sometimes agents are modeled only superficially with arbitrary encoded decision rules. The fundamental difference in the reinforcement learning setting is that only mild assumptions on the agent's behavior are encoded, as for example, they are intending to maximize their profits given their available production capacity. The learning on how to operate in a certain environment is then obtained by a combination of randomly exploring the action space and learning from past experiences on what action performed well under what circumstances. 
Hence, in reinforcement learning an agent's behavior is not predefined but arises emergently from the interplay of environment, agents, and reward mechanism. %Another related field is game theory that typically relies on best response correspondences in combination with fixed-points theorems to find stationarity points (equilibria). The fundamental difference of game theoretic solution approaches to the reinforcement learning framework is that we do not restrict strategies to be strict best responses in every iteration. In fact, even final outcomes may only be mutual best responses in a wider sense accounting for the learning history and a wide range of parameters will determine the learning quality and therefore also the convergence to meaningful outcomes.

At the metalevel, the reinforcement learning approach to find stationary points or equilibria seems to be a natural way of modelling behavior in complex environments. Given the inherent uncertainty of how competitors will operate in a certain environment, reinforcement learning may resemble the actual behavior of how humans or (ro)bots would sequentially figure out how to interact (compete) with each other. 
In particular, the most recent generation of deep reinforcement learning algorithms, has been successfully deployed to master strategic games such as
%Atari, \citep{Atari2015DQN, ScAn20},
Poker \citep{sandholmpoker2019}, Go \citep{Go2016, ScAn20}, or Starcraft \citep{Starcraft2019}, and it appears that these algorithms consistently outperform humans in playing these games. 
However, to the best of our knowledge the current wave of novel deep learning algorithms has not yet been applied to study strategic behavior in auctions-based market such as e.g., electricity markets.

The purpose of this paper is to apply state-of-the-art deep reinforcement learning to learn strategic bidding in capacity constrained uniform price auctions.
%This scenario is of high relevance as it models many countries electricity market design \cite{EUmarketdesign2019}.
Specifically, we employ the deep deterministic policy gradient (DDPG) algorithm to learn agent's behavior \citep{lillicrap2015continuous}. This algorithm is capable to tackle fully continuous state-action spaces. It is well-known that deep learning algorithms, to realize their full potential in any given problem domain, rely on a good choice of learning parameters. Therefore, this work aims to facilitate the understanding of the parameter settings critical for the performance of DDPG in uniform price auctions. We focus on a relatively simple economic problem that has an analytical solution at least in the static case which we use as a benchmark solution. More precisely, we analyze the effect of learning parameter settings on outcomes in a simple symmetric unconstrained Bertrand duopoly, i.e., where two firms compete in prices for a fixed level of demand. The Bertrand model can be analyzed in a one-stage setting following the argument that both firms would compete each other down to their marginal cost. However, it could also be viewed in a two-stage setting where the second stage is the market-clearing of a uniform price auction and in the first stage firms submit their offers. Because both firms are assumed to have unlimited capacity the auction-clearing will award the low-bidding firm to serve all the demand.\footnote{The auction mechanism is equipped with tie-breaking rules that determine what happens if both firms offer at the same price. In our case, both market participants will serve half of the demand in case of a tie.} The advantage of this setting is that it paves the way for more realistic cases as for example when firms are capacity constrained which will be our second case study. Needless to say, more sophisticated modeling assumptions, for example, accounting for firm's production portfolios (multi-unit auctions), sequential market interactions,\footnote{Short-term wholesale electricity markets are typically organized sequentially starting with the day-ahead market and allowing market participants to update their positions until (close to) real-time. For more details, see, e.g., \cite{ItRe16, GrQuIncDec20}, or \cite{GrQuCovid20}.} or including dynamic constraints that can be relevant in electricity market settings where conventional generators have non-convex production functions,\footnote{See for example \cite{Regu14}, \cite{JhLe20}, or \cite{GrQu20} on the importance on these issues.} are relatively easy to integrate into this framework. However, the purpose of this paper is to first understand the effect of the parameter settings in the reinforcement learning framework on outcomes. Therefore, we constrain ourselves at this time to study the two simple cases which we can compare to analytically derived equilibrium outcomes. Employing our developed algorithms to more complex and realistic problems in the electricity market context is ongoing work.

Our main findings are that 
(i) DDPG is a reliable tool to study competition through auctions in both capacity constrained and unconstrained settings, 
(ii) DDPG improves upon established $Q$-learning and policy gradient methods due to neither requiring discretization nor strong behavioral assumptions,
(iii) convergence rates can reach up to 99\% given optimal parametrization,
(iv) convergence depends critically on choice of memory buffer size and noise decay rate for the unconstrained case,
(v) convergence depends critically on choice of normalization method, learning rates and memory model in the capacity constrained case,
(vi) both, the choice of normalization method and memory model, impact the eagerness of agents to engage in competition and hence have to be seen as a modelling choice, even if not imposing strict decision rules. Our algorithm is open-source and accessible publicly through our Github repository.\footnote{\url{https://github.com/ckrk/bidding_learning}.} 

\cite{asker2021} and \cite{AER2020} are two recent contributions that pursue a similar goal, that is, to evaluate the effect of parameter setting on equilibrium outcomes in reinforcement learning methods. Both papers, deploy $Q$-learning \citep{watkins1989} methods which rely on a discrete state and action space. Furthermore, because all the outputs of different input (space, action) combinations must be stored in a so-called $Q$-table, discretization will always be coarse in practice given the limitations of computational resources. %\cite{asker2021} allows 100 distinct price-offers, however does not consider any state-information. 
%\cite{AER2020}'s action-space consists of 15 possible price-offers per player, with a single round of memory as state information this already leads to $15^3=3375$ distinct state-action combinations to learn. 
Instead, we use a deep-learning framework that allows us to approximate the $Q$-table and therefore enables us to work on a continuous state and action space. Moreover, we believe that scaling to larger numbers of features and players will not be possible for methods that exhibit exponential scaling in the number of features such as $Q$-learning, while deep reinforcement learning has been shown to work in complex environments.
An interesting comparison to our work in the electricity market context is \cite{lagopoplavskaya2020}, who deploy an approximate algorithm called fitted $Q$-iteration. Albeit not being a deep reinforcement learning approach in the strict sense, they use a similar approximating strategy that seems to allow them to use more complex state-spaces.

The remainder of the paper is structured as follows. 
In Section \ref{Subsec:GTtoLearn}, we provide a brief exposition of the established methods to represent strategic interactions including the discussion of analytical and computational methods. In Section \ref{sec:RL}, we give an in-depth discussion of reinforcement learning's recent development, that can be said to have experienced a recent paradigmatic shift towards deep reinforcement learning. We discuss both classical reinforcement learning, exemplified by $Q$-Learning, and deep reinforcement learning, exemplified by DDPG. In Section \ref{Sec:BenchmarkingScenario}, we start with an analytical analysis of our main benchmarking scenario, where we derive its range of static Nash Equilibria. Our benchmark scenario contains both the capacity constrained and unconstrained case. We go on in Section \ref{Sec:DDPG} to evaluate DDPG's performance relative to its known analytical solution. We accompany convergence results, with a thorough investigation of the involved learning parameters and variational analysis of the parameters. We give recommendation on choices of learning rates, memory, normalization, memory buffer size and noise decay rates. We wrap up our results in Section \ref{Sec:conclusion}.

\section{How to model strategic interactions? From game theoretic equilibrium models to reinforcement learning}
\label{Subsec:GTtoLearn}
Game theory (GT) aims to predict outcomes where players interact strategically with each other. The theories scope is wide ranging from economic applications, over biology to computer science.
Nonetheless, GT originated from the systematic study of various card and board games such as Poker and Chess. In a strikingly similar development, 
competitively playing algorithms 
for computer strategy games are one of the hallmark successes of reinforcement learning. 
The similarities are strengthened by some theoretical works that reformulate training problems into specific types of games \citep{schuurmans2016deep}.
This leads us to conclude that both fields have significant overlap despite there completely different methodological approach, thus making their comparison and interesting field of study.
In the following, we briefly review key ideas of the two distinct fields and aim to asses how both fields can complement each other.

\subsection{From Game Theory ...}
\label{Subsec:GT}
In order to predict outcomes in a strategic interaction among players GT relies on the notion of equilibrium. A game with a unique pure strategy equilibrium is regarded as strong evidence, that players will eventually reach this outcome. An equilibrium is constituted of strategies that are mutual best responses to each other, i.e., no player has an incentive to marginally deviate from her strategy given that the other players are playing their equilibrium strategies. In simple games, finding pure strategy equilibria involves typically a two-step procedure: (i) derive best responses of each player, and (ii) find set of overlapping best responses if they exist. More sophisticated games may be solved analytically (e.g., by deriving best response functions and finding their interaction) or they can be solved numerically involving non-convex optimization techniques \citep[see, e.g.,][]{GrWo20}. Best responses always exist but they are not necessarily unique and in fact best response functions can be complicated objects, that do not necessarily allow to derive equilibria.

The hardness of deriving equilibria depends strongly on the specificities of a game's formulation. Generally, derivable equilibria tend to be relatively common in single-round games. For instance, a single round of Cournot competition or Bertrand competition with differentiated products allow non-trivial closed form solutions if the functional forms defining costs and demand are favorable. However, adding dynamics to the one-shot game poses significant challenges. In infinite games or finite games with uncertain termination point meta-strategies, such as  e.g punishments of non-cooperators,---a strategy which would not make sense in a one-shot game---arise. The issue is even further complicated by the fact that the efficiency of such meta-strategies depends critically on model parameters such as  the discount rate which in practice is hard to measure. A whole class of theorems so-called ``folk-theorems'' derive results that can essentially enforce any of a game's static equilibria a
% agent based model lit: https://pubs-aeaweb-org.stanford.idm.oclc.org/doi/pdfplus/10.1257/aer.102.3.53, % https://www-jstor-org.stanford.idm.oclc.org/stable/2171879?seq=1#metadata_info_tab_contents,
%https://www-aeaweb-org.stanford.idm.oclc.org/articles?id=10.1257/jel.20191434, https://pubs-aeaweb-org.stanford.idm.oclc.org/doi/pdfplus/10.1257/jel.20191434
%https://arxiv.org/pdf/2102.11107.pdfs an equilibrium of the entire extended game depending on the value of the discount rate \citep{folkthm1986}. The situation is mirrored in related relevant disciplines such as auction theory. For instance, optimal bidding behavior within a pay-as-bid auction is solely determined by a players belief about other players valuations or behavior \citep{krishna2009auctiontextbook}. Similarly, for our use case more relevant results, hold true for marginal players in uniform price auctions. Here belief priors replace discount rates as critical but non-available parameters.
This makes the assessment of sequential strategies strongly dependent on predefined modelling choices which may impact the overall quality of the predicted outcomes. 

In fact, only highly stylized cases support unique equilibria for extended games. More commonly, we face situations with either infinitely many equilibria or no pure strategy equilibrium at all. Furthermore, for mixed equilibria existence results are known, but not necessarily a constructive way to compute them.%GT frequently addresses these issues through the application of fixed-point theorems that deliver existence and uniqueness results of equilibria, but do not yield explicit constructions of the involved equilibria.

The lack of constructive equilibria for repeated games is an unsolved problem within GT. Due to the abundance of equilibria as predicted by folk theorems, this is strongly related to the question of equilibrium selection. Consequently, many heuristic approaches for equilibrium selection exist. Furthermore, refinements of the Nash equilibrium concept, such as e.g., the requirement of sub-game perfection, are often used to reduce the number of equilibria. Alas, in the frequent case of non-unique equilibria, there is no unified theory of equilibria discrimination that allows to decide which equilibria are more likely or advisable to be played. A systematic resolution of the issue does not seem imminent and is one of the key motivations for us to look beyond game theory towards solution concepts from outside of game theory in order to model mid- to long-term strategic interactions.
 
\subsection{... over Algorithmic Game Theory ...}

Classic Game Theory typically concerns itself with analytical derivations, existence, and uniqueness results of Nash Equilibria. It usually does not deliver a theory of suboptimal and out-of-equilibrium play. Out-of-equilibrium play occurs prior to all players arriving in a given equilibrium and external influences may require frequent ``reequilibrations.'' GT usually does not address whether the time scales spent in such an out-of-equilibrium play are short or long. Algorithmic game theory (AGT) addresses this issue by analyzing whether a given Nash equilibrium can be computed in finite or even short time scales \citep{roughgarden2010algorithmic}. 

For instance, this is already illustrated by Algorithmic Game Theory's analytical results that particular learning behavior assumptions are able to efficiently learn specific subsets of equilibria \citep[i.e., best response-dynamics converge to correlated equilibria, see e.g.,][]{BRconvergence2008,foster1997calibrated}, no-regret-dynamics converge to coarse correlated equilibria \citep{blum2008noregretdynamics}. Moreover, there are results that even Nash equilibria can be found efficiently for specific types of games such as potential games \citep{viossat2013}. Thus AGT adds the important quality criterion of computational efficiency to GT's toolbox. 

However, algorithmic game theories results stems from theoretical computer science with a complexity theoretic flavour, where frequently theoretical algorithms are used to reason about existence and no-go results. This means that AGT's notion of computational efficiency usually equals the complexity theoretic notion of polynomial time computable \citep{roughgarden2010algorithmic}.
AGT's algorithms are not necessarily always practically efficient or even constructive. As such AGT's approach refines the fixed-point theorem approach, but too offers constructive results only in particular cases.

Furthermore, we are not aware of comparably strong results for bidding games as we have at our disposal in the case of potential games. We therefore hope to supplement these results with the use of concrete state-of-the-art deep learning algorithms that are considered practically successful in the Reinforcement Learning community and critically asses their viability for learning equilibrium strategies.

\subsection{... to Reinforcement Learning}

GT relies on the concept of a Nash-Equilibrium to predict a given player's behavior. An alternative solution concept could be to simulate a players behaviour by providing a learning algorithm that extrapolates from past rounds to future rounds and recommends the next move. This move is not necessarily an optimal move, but it is the move that appears optimal given a certain set of experiences and learning parameters.

GT relies on mathematical optimization and the theory of metric spaces to identify mutual best responses, i.e., equilibria. Reinforcement learning (RL) essentially is a form of educated Monte-Carlo simulation, where randomly simulated moves are used as a basis for a statistical inference that determines what strategy is the best conditioned on a players history.

There exists a certain trade-off between both approaches here. More precisely, reinforcement learning algorithms can by design give no guarantee or certificate of optimality. However, reinforcement learning is versatile and does not require specific functional forms, like convexity, of the constraints or objective function, that are common in mathematical optimization. Learning algorithms can either converge naturally or convergence can be enforced trough hyper-parameters. The quality of such solutions may vary, especially if convergence is enforced. Nevertheless, once convergence is attained, a strategy is always singled out.

Moreover, the unresolved question of equilibrium selection is side-stepped by the use of learning algorithms. It could be argued that the equilibrium selection criterion of RL is in fact the learnability of a given equilibrium. Furthermore, due to reliance on Monte-Carlo simulation several runs of an algorithm may result in differing solutions. Dealing with variation within solutions of RL algorithms is certainly a challenge, but broadly, this can be seen as an analogue to the equilibrium selection problem encountered in GT.

\section{Modern reinforcement learning: From the $Q$-table to deep neural networks}
\label{sec:RL}

We devote the following section to a discussion of recent advances in the machine learning domain with a focus on on algorithms that we deem best suited to learn strategic bidding in uniform price auctions.

Machine learning is subdivided in supervised, unsupervised, and reinforcement learning. Unsupervised learning is completely independent of the problem structure making it versatile but also somewhat limited in scope. It is usually used rather for classification tasks in not fully understood problem domains and not fit for strategic analysis. Supervised learning in-turn relies on human experts to train algorithms. In supervised learning, experts discriminate exemplary ``training'' actions into good or bad. 
The algorithm then learns to mimic the provided discrimination into good and bad and is as such not independent of the modellers input. In our opinion, the third option, RL, is the most interesting approach for tackling strategic games. It relies on an exact specification of an reward mechanism, that trains players that seek to increase their respective reward. In many markets that rely on auctions---electricity markets are no exception to that---the market-clearing mechanism is known. That means, for an exogenous set of offer and bid curves for each market participant, the market-clearing results can be easily computed. The challenge though, is to derive how market participants will formulate the offer curves they submit to the market. Remember RL algorithms requires only to specify, i.e., hard-code, the reward mechanism. Given that, equilibrium offer and bid curves may be learned by iteratively submitting curves to the market-clearing mechanism and receiving a response, i.e., a per-round profit value conditional on my submitted curve and the curves submitted by my opponents.

\subsection{Reinforcement Learning}
We briefly review reinforcement learning's key concepts \citep[see, e.g.,][for additional details on the concept]{sutton1998}.%buch upgedated, habe linkref gleich gelassen

Reinforcement learning's goal is to determine an optimal action  given a certain state, while drawing only on the information of a sequence of past rewards.
In order to capture the inherent time-structure of RL we define a time index $t$.
Let us write the action at time $t$ as $a_{t}$, the state at time $t$ as $s_{t}$, and the reward received at time $t$ by
$R(a_{t},s_{t})$ depending on $a_{t}$ and $s_{t}$.
We do not specify $a$ and $s$ closer by intention, since their exact meaning depends heavily on the problem at hand. However, we point out that they are not necessarily integers or even single numbers, they could be vectors or chosen from an continuous interval. Similarly, infinite time sequential optimization tries to maximize a value

\begin{equation}
\label{eq:Value_over_time}
    V(s_{0}):= \sum_{t=0}^{\infty}\max_{a_{t}} \gamma^{t} R(a_{t},s_{t})
\end{equation}

that is essentially a time averaged discounted objective function, where $0 \leq \gamma < 1$ is a discount factor required to enforce \eqref{eq:Value_over_time}'s convergence. In infinite time sequential optimization, there are various tools such as dynamic programming to explicitly evaluate \eqref{eq:Value_over_time} if possible. Nonetheless, this can be a non-trivial task.
RL shares the same aim as the field of mathematical optimization, but can be considered an alternative methodological approach. Hence, the goal of RL is to converge towards

\begin{equation}
\argmax_{a} R(a,s),
\end{equation}

for all $s$ or at least approximately find $a_{t}$ such that

\begin{equation}
 \sum_{t=0}^{\infty}\gamma^{t} R(a_{t},s_{t}) \approx V(s_{0}).
\end{equation}

In RL an action is determined from a sequence of rewards and actions $a_{0}, \ldots ,a_{t}$ taken in specific states $s_{0}, \ldots ,s_{t}$ by means of a learning rule. Sequential optimization problems can be decomposed by means of the (infinite) Bellman equation as

\begin{equation}
\label{Eq: Bellman}
V(s_{0}):= \sum_{t=0}^{\infty}\argmax_{a_{t}} \gamma^{t} R(a_{t},s_{t})=
R(a_{0},s_{0}) +
\sum_{t=1}^{\infty}\argmax_{a_{t}} \gamma^{t} R(a_{t},s_{t})
=
R(a_{0},s_{0}) + \gamma V(s_{1}).
\end{equation}

Solving \eqref{Eq: Bellman} involves several well-known challenges.  First, infinite time horizon problems do require a so-called discount factor $\gamma$ to become solvable. The important modelling choice of $\gamma$ is common to all approaches relying on \eqref{Eq: Bellman} and not limited to RL at all. Similarly, game theoretic folk theorems \citep{folkthm1986} encounter the problem of solutions relying sensitively on the choice of $\gamma$. Second, if no closed form solution is known for \eqref{Eq: Bellman} computationally solving an infinite expression is not feasible. Therefore, it is common practice to solve relaxed versions of \eqref{Eq: Bellman} to obtain approximately optimal solutions. In essence, all reinforcement learning methods are specifications how to approximate $V(s_t)$ in \eqref{Eq: Bellman} from a given history of $t$ actions and states.

\subsection{$Q$-Learning}
The main model-free approach to estimate $V$ is the family of temporal difference learning algorithms. Its most well-known member is so-called Q-learning \citep{watkins1989}.
$Q$-learning's  core assumption is that the state-action space is discrete and that the discounted cumulative reward function $V(s_{t})$ can thus be represented as a table or matrix whose values contain the possible rewards originating $R(a_{t},s_{t})$ from any combination of $S$ states and $A$ actions. If such a table were available, finding optimal actions for a state reduces to reading of the maximal value in a column corresponding to the state of interest.

$Q$-learning's key strategy is to find such a matrix representation of $V(s_{t})$  through an update rule

\begin{equation}
\label{eq:qlearning_update}
Q_{t+1}(s_t,a_t) = 
(1-\alpha) Q_t(s_t,a_t) + \alpha [R(s_t,a_t) +\gamma \max_{a}Q_t(s_{t+1}, a)],
\end{equation}

where $Q$ is an initially arbitrary matrix of size $S \times A$, $\alpha$ is the so-called learning rate and $\gamma$ the discount factor. Repeated application of \eqref{eq:qlearning_update} will eventually lead to convergence towards $R(s_t,a_t) +\gamma \max_{a_t}Q_t(s_{t+1}, a_t)$ if $R$ remains stationary and all values of $Q_t(s_t,a_t)$ are updated alike. $R(s_t,a_t) +\gamma \max_{a_t}Q_t(s_{t+1}, a_t)$ is the first order approximation of the Bellman equation \eqref{Eq: Bellman}, hence under the assumptions of stationarity, ergodicity, and sufficiently many iterations the matrix $Q$ is expected to converge towards an approximation of $V$.

$Q$-learning usually selects the action that has the maximal value in the $Q$-table, but ergodicity is only ensured by sometimes playing a random move instead (chosen with probability $\epsilon$ the exploitation-exploration parameter). Typically, $\epsilon$ is decreased throughout the run-time of the algorithm until it reaches a lower threshold to allow for easier convergence the algorithm has run through many iterations. The speed of the decrease in $\epsilon$ is termed the decay rate.

Under ideal conditions $Q$-learning can thus reliable approximate the Bellman equation. However, $Q$-learning's fundamental prerequisites can be hard to ensure in practice. For instance, stationarity of rewards is a strong assumption in general, but also particularly when studying auctions. $Q$-learning is traditionally used by single agents. Especially, in the multi-agent case relevant for auctions, competing agents essentially form a dynamic state-space that makes it hard to ensure stationarity \citep{MARL2016learning,MARL2008}. Moreover, attaining ergodicity and short runtimes becomes harder to achieve the larger the state-action space gets. This means that fine discretizations of the state space leads to unfavourable scaling in the number of features and can lead to humongous state-action spaces even in relatively simple models. Furtherore, memory requirements constrain modelers to use rather coarse model descriptions in practice (see Table~\ref{Fig:RLstatespacesizes}). Aside from scaling considerations, we want to close this sections with a cautionary tale stemming from theoretical considerations in Operations Research. It is self-evident that discretization of state-spaces introduces rounding errors. Naturally, a discrete and continuous algorithm will be operating in slightly different scenarios due to their respective action spaces.
Hence, $Q$-learning will in-fact play a slightly distorted discretized version of the more general continuous game. The fineness of the discretization directly influences the number of ensuing equilibria. For instance, even straight-forward discretized Bertrand games are examples that exhibit 2--3 equilibria depending on the choice of discretization, while the continuous Bertrand game is actually unique \citep{asker2021}.

Sometimes discretization is justified with the perceived negligibility of small rounding errors. Nonetheless, discretization alters the underlying problem fundamentally. For instance, a discretized linear program becomes an integer linear problem whose solution can be arbitrary far away from the originally continous problem regardless of the possibly small magnitude of the rounding errors \citep[see IP Myth 1 in][]{ORmyths}. Hence, discretization is not necessarily a stable operation and can alter a problem's solution which may bias $Q$-learning's outcomes.

\begin{table}[ht]
\centering
\small
\caption{Comparison of Recent Articles applying $Q$-learning}
\begin{threeparttable}
  \begin{tabular}{lp{2.9cm}p{2.9cm}p{2.9cm}p{2.9cm}}
  \toprule
          & \cite{QlearningAncillary2020} & \cite{AER2020} & \cite{asker2021} & \cite{aliabadi2017agent}  \\
   \midrule
    Domain & Sequential Electricity Markets (Multi-unit auctions) & Bertrand Duopoly & Bertrand Oligopoly & Electricity Markets (Multi-unit auctions) \\ \addlinespace[5pt]
    Action Space Scope & Single Price-offer per unit  & Single Price-offer       & Single Price-offer      & Single Price-offer        \\ \addlinespace[5pt]
    Action Bounds & $\left[ 0, 10^{2} \right]$     &    Bertrand- to Monopoly prize   &$\left[ 0.1, 10 \right]$       & $\left[ 9, 45 \right]$        \\ \addlinespace[5pt]
    Action Stepsize & 1      & Scale-dependent       & 0.1       & 9-10       \\ \addlinespace[5pt]
    Action Space Size & $10^2$      & 15      & $10^{2}$      & $3\cdot4\cdot4=48$         \\ \addlinespace[5pt]
    State Space Scope & Marginal price, volume weighted average price, and total demand\tnote{1} & Single round Memory              & None       & None  \\ \addlinespace[5pt] 
    State Stepsize &  1     &  Scale-dependent     &    None   & None         \\ \addlinespace[5pt]
    State Space Size & $10^2$ ($\cdot$ Demand)& $15^2$       & 1      & 1        \\ \addlinespace[5pt]
     State-Actionspace Size & $10^4$      & $15^{3}$      & $10^2$       &  $48$       \\ \addlinespace[5pt]
    \bottomrule
    \end{tabular}%
    \begin{tablenotes}
        \footnotesize
            \item[1] Demand is only used in few scenarios with unspecified granularity. Graphs indicate possible granularities of 5 or 100.
        \end{tablenotes}
   \end{threeparttable}
  \label{Fig:RLstatespacesizes}%
\end{table}%

Table~\ref{Fig:RLstatespacesizes} gives a non-exhaustive overview of very recent reinforcement learning implementations of multi-agent reinforcement learning in economic settings. Table \ref{Fig:RLstatespacesizes} compares the relation between the papers area of study and the granularity of representation.
\textit{Domain} summarizes what scenario the paper describes. \textit{Action and State Space Scope} details how states and actions are modelled. 
Naturally, modelling choices have to be made. We describe those both for action and state spaces with respect \textit{Stepsize, Bounds and Space size}.
Bounds denote upper and lower limit of the modeled space. Size is the number of elements in the space. Finally, stepsize denotes the distance between individual elements.
We employ the convention, that we write a given states size as product of its individual constituents (i.e., $10 \cdot 10=10^{2}$ labels a space containing two distinct features with 10 possible values each).
Whether a given number of state-action pairs appears sufficient, depends on the desired accuracy and modelling domain.
The employed state-action spaces tend to reside in the order of $10^{-1}$--$10^{1}$ \citep{QlearningAncillary2020,AER2020}. Despite the relatively large magnitude of the composite state-action space, the resolution of indivdiual features remains relatively modest typically ranging from 1--10 units, unless trivial state spaces are involved. The number of distinct features (i.e., actions or state space elements) usually is less than between 3.
However, even slightly older papers can be found that exhibit action spaces of less then 100 elements \citep{aliabadi2017agent}.

Overall, we have discussed several challenges experienced within classical reinforcement learning exemplified by $Q$-learning.
Not all of these challenges can be expected to be easily overcome by deep learning.
For instance, non-stationary state spaces are a recognized problem of the notoriously hard multi-agent games that seem unavoidable from an algorithmic perspective and remain common to all reinforcement learning methods be it classical or deep. 

However, other issues can very well be adressed by deep learning. Consider, either $Q$-learning's problematic state-action space scaling or its possible discretization errors. Both are intimately tied to $Q$-learning naturally discreteeness and can be remedied through the application of modern deep learning methods. Although, there exist deep learning approaches for discrete problems as well \citep[for instance, DQN,][]{Atari2015DQN}, we will go on to elaborate on the possibilities of natively continuous algorithms in the following section.

\subsection{Deep Deterministic Policy Gradient}
\label{subsec:ddpg_lit}
Deep Deterministic Policy Gradient (DDPG) extends $Q$-learning by two key concepts: 
\textit{neural networks} and \textit{actor-critic} design.

First off, when opting for a neural network based design, we forfeit the idea to explicitly storing information about any possible state action combination.
This is desirable whenever the state-action space grows larger.
It is true, that some problems are well-representable with coarse discretizations and these are exactly those where $Q$-learning will perform well. Alas, many naturally occurring quantities such as temperature or quantities of energy are inherently continuous parameters. Even, essentially discrete parameters such as prices may require very fine discretizations if adressed in full-generality (i.e., down to cent scale). For such problem domains, neural network based designs allow to represent continuous as well as almost continuous domains easily.

If $Q(s_{t}, a_{t})$ grows too large to memorize and process efficiently, one could be satisfied with finding a parametrized family of functions that guarantees 

\begin{equation}
    \label{eq:NNapprox}
    \min_{\omega_{1},\ldots,\omega_{p}} \sum_{t}||C(\omega_{1},\ldots,\omega_{p},s_{t},a_{t}) - Q(s_{t},a_{t})||_{2} < \delta,
\end{equation}

for a non-linear function $C$, a sufficiently small $\delta$ and a fixed set of parameters $\omega_{1},\ldots,\omega_{p}$, where we require $p << S \cdot A$ (i.e., there are significantly less weights required, than needed for an exhaustive statespace description).
We call this neural network $C$ the \textit{Critic} network and $\omega_{1},\ldots,\omega_{p}$ its \textit{weights}.
\eqref{eq:NNapprox} states that $C$ is an approximation of the $Q$-table $Q$, that is in turn an approximation of the problems discounted cumulative reward function $V$ obtained from the Bellman equation.
This would allow us to store comparable amounts of information with much less parameters.
In principle, many function families can attain this, however neural networks have been proven to always fulfill \eqref{eq:NNapprox} if $p$ is large enough
by  the so-called universal approximation theorem \citep{Cybenko89}.

This allows us to make the approximation
\begin{multline}
\label{eq:CriticApprox}
(1-\alpha) Q(s_t,a_t) + \alpha [R(s_t,a_t) +\gamma \max_{a}Q(s_{t+1}, a)] 
\approx \\
(1-\alpha) C(\omega_{1},\ldots,\omega_{p},s_t,a_t) + \alpha [R(s_t,a_t) +\gamma \sup_{a}C(\omega_{1},\ldots,\omega_{p},s_{t+1}, a)],
\end{multline}
where we do only memorize a finite number of parameters $\omega_{1},\ldots,\omega_{p}$ and not necessarily the possibly many $a_{t}$ and $s_{t}$ arising in complex state-action spaces.
Moreover, a function allows for continuous outputs, thus paving the way for continuous state-action spaces.\footnote{Note, that \eqref{eq:CriticApprox}'s RHS is a proper generalization of its LHS, in the following sense:
The LHS is only sensible for discrete values of $a$.
In contrast, the RHS approximation of \eqref{eq:CriticApprox} makes sense for continuous $a$, but is equally sensible for discrete $a$. Consequently, both variants are applied within RL. 
If \eqref{eq:CriticApprox}'s $a$ is taken discrete the algorithm is called ``Deep $Q$-Learning (DQN)'', while for continous $a$ we call it a ``Deep Deterministic Policy Gradient (DDPG)'' algorithm. }

However, a further complication arises when moving to continuous action spaces. The evaluation of \eqref{eq:CriticApprox}'s left hand side is straight forward. In particular, it is easy to evaluate the maximum $\max_{a}Q(s_{t+1}, a)$ since $Q$ is just a finite table. In $Q$ each row (or columns) is simply a finite list, that can efficiently be sorted to retrieve its maximal element.
%we can simply compare each entry to find the largest one. 
In contrast, if $a$ is interpreted as a continuous input parameter of the neuronal network $C$, we are faced with
$\sup_{a}C(\omega_{1},\ldots,\omega_{p},s_{t+1}, a)$, where we have to take a supremum over infinitely many values of $a$.
This is a non-trivial optimization problem instead of simple list sorting.

In order to circumvent this problem, we introduce another neural network $A$ that maps from the state-space to the action-space and may depend on another set of weights $\omega'_{1},\ldots,\omega'_{r}$. We call $A$ the \textit{Actor network}. 
This is only one possible design choice to solve the issue, but it is the hallmark of so-called \textit{actor-critic methods}.
Instead of a neural network one could assume a functional form of $A$ that allows for efficient solution of the maximization problem, but that would be more similar to so-called policy gradient methods. The idea of the neural network based approach is that by updating

\begin{equation}
\sup_{\omega'_{1},\ldots,\omega'_{r}}\sum_{t} C(\omega_{1},\ldots,\omega_{p},s_{t},A(\omega'_{1},\ldots,\omega'_{r},s_{t}) )
\end{equation}

iteratively the actor becomes an approximation of the function mapping states to ideal moves, thus

\begin{equation}
    \sup_{a}C(\omega_{1},\ldots,\omega_{p},s_{t}, a)
    \approx
    C(\omega_{1},\ldots,\omega_{p},s_{t}, A(s_{t})).
\end{equation}

Overall, we simplify \eqref{eq:CriticApprox} to

\begin{equation}
(1-\alpha) C(\omega_{1},\ldots,\omega_{p},s_t,a_t) + \alpha [R(s_t,a_t) +\gamma C(\omega_{1},\ldots,\omega_{p},s_{t}, A(s_{t}))],
\end{equation}

leaving us with the actor-critic update rule. If the actor network is chosen to output a single deterministic value  $a_{t}$ for each input state $s_{t}$, we obtain the update rule that is characteristic for a deep deterministic policy gradient (DDPG) learning algorithm.

\section{Price-competition in a Symmetric Capacity Constrained Duopoly}
\label{Sec:BenchmarkingScenario}
In order to test the performance of our algorithms to derive Nash equilibria in a strategic bidding context, we first discuss our benchmark scenario. We deliberately chose a simple framework that has an analytical (static) solution but at the same time captures the nature of price competition in a sealed bid uniform price auction setting. 
Overall, we aim to asses the performance of a modern continuous reinforcement learning algorithm (DDPG). Consequently, the benchmarking scenario should have a straight forward continuous formulation.
Moreover, we do not want to limit our analysis to the relative performance of our algorithm with any particular set of competing algorithms but we desire some form of absolute performance measure. Hence, we require a benchmarking scenario that is analytically solvable. We have kept the scenario simple so that all Nash-Equilibria can be computed by hand. This allows us to comment on every equilibrium and analyze how hard it is to find respectively.

We study the following scenario: Assume two identical players, $i$ and $j$, competing against each other in prices. Both players control a capacity $\overline{q}_i=\overline{q}_j=\overline{q}>0$ allowing them to commit to produce a homogeneous product. Both players have equal marginal cost $c_i=c_j=c \geq 0$. Players participate in a uniform price auction setting by submitting price offers. We write the price offers without loss of generality as $p_i \leq p_j \leq \overline{p}$. Offer prices are constrained by a price-cap ($\overline{p}$). 
If both players were submitting the same offer prices to the market, each player may sell one-half of the total demand (``tie breaking rule'').
We furthermore, assume inelastic demand $D$, that we require to be within $\overline{q} < D < 2\overline{q}$ to ensure non-trivial solutions. 

For any given set of parameters ($\overline{q},c,\overline{p},D$) and the tie break rule as described above we can compute an interval of Nash-Equilibria.
We calculate the Nash Equilibria in Proposition \ref{Thm:NE_characterization}.
There, we find and characterize a non-trivial range of equilibria that is well above the marginal costs of the players.\footnote{This is in so-far remarkable as our Scenario is a Bertrand competition model with capacity constraints and a homogeneous product. In their classic formulation as sequential game with pay-as-bid prices, these oligopoly models are known to lack pure Nash equilibria since they can exhibit indefinite out of equilibria dynamics such as Edgeworth cycles \citep{maskintirole1988}. In contrast, in our uniform price auction setting the scenario becomes tractable through the introduction of a maximal price. We want to emphasize that the influence of the maximal price is a specific effect of the uniform price auction setting, since pay-as-bid Bertrand competition's equilibrium strategy is to offer marginal costs. Above-marginal costs are only attained due to capacity limitations in the classic pay-as-bid case and a high enough maximum price should not influence the fact that Bertrand competition tends towards the marginal costs \citep[see][]{undercutproof}.}

This allows us to contrast numerical results against a well understood background. We have conducted all our numerical simulations with the following set of parameters, $\overline{q} = 50, c = 20, \overline{p}= 100$, and $D=70$, unless specified otherwise. We emphasize that the condition $\overline{q} < D < 2\overline{q}$ is crucially needed for the existence of non-trivial Nash equilibria. Alternatively, one may consider either the ``uncompetitive'' case $2\overline{q} \leq D$ or the ``ultra-competetive'' case $D \leq \overline{q}$. 

In the ``uncompetitive'' case all capacity is sold, regardless of the players behavior. Consequently, all strategies, where at least one player bids at the price cap are within equilibrium. 
We consider this case to be a trivial case, that we do not discuss further. 

In contrast, the ``ultra-competetive'' case $D \leq \overline{q}$ essentially renders the capacity constraint irrelevant and thus reduces to standard Bertrand competition, where both players bidding the marginal costs constitutes the equilibrium strategy in this case. Although, the predictions for the standard model of Bertrand competition are well-known ($p_i^* = p_j^* = c$), this case is not necessarily trivial when the goal is to learn this outcome through a black-box method. We believe that complying with game theories predictions for the special case of (non-capacity) constrained Bertrand competition is an important minimum requirement in order to trust a learning algorithm. Therefore, we include it in our discussion of the results later on.

If we would discretize the players action space, the two player condition allows us to easily visualize the games payoffs as a bimatrix . Moreover, this discretization matches the discrete action space of $Q$-learning. Computing Nash Equilibria can be done straightforward from the bimatrix. Each best response can be red of row-wise and column-wise. We find all equilibria precisely, where row- and column-wise best responses match.
In contrast, computing the equilibria of the continuous case requires a calculation. 
Our scenario is chosen precisely such that this calculation remains easy,
however in more general scenarios we have no guarantee of analytical solvability.

\begin{proposition}
\label{Thm:NE_characterization}
Consider an uniform price auction with two participants and eligible bids that may be placed in the interval $\left[-\infty, \overline{p} \right]$.
If both participants offer a fixed quantity $\overline{q}$ of an indistinguishable product, whose production involves a constant per unit production cost of $c$ and the demand $D$ lies within $\overline{q} < D < 2\overline{q}$, then there exist w.l.o.g.\ a number of tuples of equilibrium strategies $p_{L}<p_{H}$. The range of lower equilibrium offer prices being all $p_L$ with
\begin{equation}
\label{eq:NE_charakterization_low}
    p_{L} \leq (\overline{p} - c)\frac{(D-\overline{q})}{\overline{q}} + c, 
\end{equation}
and the unique higher equilibrium offer price $p_{H}$ being
\begin{equation}
\label{eq:NE_charakterization_high}
        p_{H}=\overline{p},
\end{equation}
is the auctions range  of Nash-Equilibria.
\end{proposition}

\begin{proof}
Let us start by identifying the players best responses. We begin by case distinction, either $p_{L}=p_{H}$ or $p_{L} \neq p_{H}$ holds. In case of $p_{i}=p_{j}$ the tie breaking rule applies and each player receives $p_{L} \frac{D}{2}=p_{H} \frac{D}{2}$ Since, any player could improve his profit to $p_{H} \overline{q}$ by minimally lowering his offer and thus selling its whole capacity, bidding $p_{L}=p_{H}$ can not be a best response for any player. \\

Hence, without loss of generality we can assume $p_{L}<p_{H}$ (thus allowing us to interpret $L,H$ as low and high bidding respectively).
Let us consider the action of the high bidding player $p_{H}$. $p_{H}$ is a best response if player the high-bidder $H$ can neither gain by highering or lowering its offer. 

Since the higher player is the price setting player, we can write the profit of the high-bidding player as
$\pi_{H} = (p_{H}-c) (D-\overline{q}) \leq (\overline{p}-c)(D-\overline{q})$. Thus the higher bidding player can always increase its payoff by highering its offer, unless it bids at the price cap.
This makes bidding at the price cap the only candidate for a best response of the high bidding player. 

It remains to discuss, whether the high bidding player can improve its profit by lowering its offer. Lowering ones offer can only increase ones profit, if one undercuts the competition and thereby gains a large share of the market. This means that in order for $p_{H}=\overline{p}$ to be a best response the following equation has to hold:

\begin{equation}
\label{eq:undercut_proof}
(p_{L}-c) \overline{q} \leq (\overline{p}-c) (D-\overline{q}),
\end{equation}

otherwise the high bidding player might gain by undercutting the low bidding player. For the low-bidding player in-turn, all prices below the price cap yield the same profit and bidding at the price-cap leads to a tie that decreases its profit. Therefore, the low bidding player is too playing a best response. Hence, if equation \eqref{eq:undercut_proof} holds both player are playing a best response and thus a Nash equilibrium.

Therefore, the high bidding player is completely fixed, however the low playing player has a range of possible equilibrium strategies. To characterize these,
we go on to reformulate the above condition to

\begin{equation}
(p_{L}-c)  \leq (\overline{p}-c)\frac{(D-\overline{q})}{\overline{q}}.
\end{equation}

Now considering that $\overline{q},c,D$ and $\overline{p}$ are fixed by our assumptions, we have found a closed representation of $p_{L}$ being in equilibrium and the statement to be shown follows.
\end{proof}
We briefly comment on the implications of Proposition \ref{Thm:NE_characterization}. First, we see that the interval of equilibria depends on production costs, the ratio of sales and the price cap. While the dependence on price, demand and capacity seems very natural, the appearance of the price cap my strike an observer as unnatural. Although, almost all stock and electricity markets indeed have upper and lower offer limitations those are almost never reached in practice, thus one may question why the seemingly irrelevant price cap appears in the characterization of the equilibrium. \\

Indeed, this comes from our decision to have a very simple test scenario. In our, case the assumption of completely inelastic constant demand leads to the price cap being the only limiting factor on the bids of the high bidding players. In reality or a more sophisticated model, this limit would be set by the consumer or the demand function. Hence, the price cap should rather be seen as a proxy for a more involved demand function. Any demand function, would provide such an upper bound but at the very least this could be the monopoly price ensuring that the high bidding player is always bounded from above as required in the proof of Proposititon \ref{eq:undercut_proof}.

We want to briefly compare and contrast our result to the similar but different concept of undercut proof equilibria (UPE) \citep{undercutproof,shy2001}, where both players are playing undercut proof strategies.
These are usually found in an entirely different subject: Bertrand competition with differentiated products as opposed to the completely homogeneous product we are considering.
The concepts of undercut proof strategy and undercut proof equilibrium can be regarded as corresponding to the notions of best response and Nash equilibria respectively.
For differentiated products UPE's provide general existence and uniqueness results, that are not possible for Nash equilibria.
Our setting is different though and the similarity in methods are a coincidence.
There is a striking similarity between equation (\ref{eq:undercut_proof}) in the proof of Proposition \ref{Thm:NE_characterization} and the UPE, since  (\ref{eq:undercut_proof}) enforces that the high-bidding player has no incentive to offer lower than his lower bidding opponent. Hence, the low-bidding player plays an undercut-proof strategy that is simultaneously his best response. Conversely, the high-biding player is not playing undercut-proof strategy, but his best response instead. This means that in our benchmarking scenario from Proposition \ref{Thm:NE_characterization} the set of strategies (\ref{eq:NE_charakterization_low}) and (\ref{eq:NE_charakterization_high}) are not an undercut proof equilibrium but indeed a true Nash equilbirium.

\section{DDPG-Learning in Uniform Price Auctions}
\label{Sec:DDPG}

Deep Deterministic Policy Gradient learning (DDPG) is a successor of classical RL ($Q$-learning) allowing for continuous action and state spaces. Clearly, storing all function evaluations in a $Q$-Table becomes infeasible in a continuous setting. Therefore the $Q$-table is only approximated using neural networks. DDPG employs a so-called Actor-Critic design. Hence, actions and states qualities are evaluated by two differing neural networks: the Actor network and the Critic network respectively. For example in the bidding context, the actors submit offer prices and the critics collect the offers and past rewards in order to predict future hypothetical profits. 
Once all offers are submitted, the market clearing determines winners and losers of the auction and awards profits to the agents.
Agents process this information and update their neuronal networks to incorporate the new information gained by the current auction. 
%Aside from the differing network designs discussed in Section \ref{subsec:NNdesign}, the most relevant design choice of a neural network is the learning rate.
%Hence, DDPG relies on a so-called actor-critic design where each agent is equipped with two neural networks: the actor and the critic, where the actors perform actions and the critics judge these actions. 
%The idea is basically that both neural networks evaluate differing parts of the Bellmann equation in continuous settings.
%The critic essentially extrapolates past reward data to future continuous state-action pairs. Given the critic's reward predictions the actors replaces the maximization over a continuous action space that is not easily performed without a predefined functional form describing all possible actions (c.f. Section \ref{subsec:ddpg_lit} ).

DDPG is a well-tested and successful algorithm, that has performed well in control tasks in a range of continuous environments, (e.g., several robot motion control tasks \cite{zhang2019DDPG}, autonomous driving \cite{DDPG2018driving}).
In this section, we proceed to demonstrate that DDPG can also learn strategic behavior in continuous environments such as uniform price 
auctions. Moreover, we give a detailed presentation of our DDPG implementation and discuss the impact of relevant design choices such as learning rate, memory and normalization methods.

We deploy DDPG to explore its efficacy to find equilibria in a competitive bidding environment. There is not yet a standardized protocol that assess these complex algorithms and one of our contributions is a first framework to perform such an analysis. Furthermore, we regard DDPG as a particularly successful instance of the much larger family of deep learning algorithms. We believe that our results can be taken as exemplary for treating the much broader question of \textit{deep learning's efficacy for simulations of strategic bidding behavior}.

In order to test DDPG, we rely on the simple but non-trivial benchmarking example described in Section~\ref{Sec:BenchmarkingScenario}, with an analytical solution. This is important because it allows us to classify whether we successfully learned an equilibrium outcome. We interpret convergence towards an equilibrium outcome as a sign of successful learning. 
Finally, we perform statistics on the strategies the algorithm converged to after a fixed learning period. The basic benchmark scenario is kept constant in order to evaluate the influence of several parametrizations of the algorithm.

One of the main challenges of assessing reinforcement learning algorithms in general is their dependence on a number of user-defined ``hyper-parameters.'' Already for simple $Q$-learning, parameters such as learning rate, discount rate, and exploration-exploitation parameters will affect the algorithm's speed and convergence. When approximating the $Q$-function by a neural network, additional parameters, such as the memory buffer size, need to be chosen. The DDPG framework consists of two neural networks, hence, all the aforementioned parameters will affect solution quality and solution performance but the benefit of this approach is that it is able to treat continuous state-action spaces.

There is no generally agreed prescription on how to select hyperparameters, however, specific problem domains have distilled their experience into broad design guidelines for selecting these parameters. For instance, image recognition has nowadays been associated with convolutional neural network architectures (CNN's) and image generation is typically relying on generative adaptive networks (GAN's). Employing this kind of methods to market equilibrium problems is relatively new and therefore we can not draw on any such prescriptions. Our design philosophy, is to closely mimic the parametrization of the seminal paper that introduced continous actions-space deep learning \citep{lillicrap2015continuous}. We precisely list all relevant hyperparameters that are related to deep learning in Section~\ref{subsec: parameter description}. \cite{lillicrap2015continuous} uses several environments from the Multi-Joint dynamics with Contact (MuJoCo) library that simulates movement of rigid bodies, that are similar to the problems encountered in robotics. In Section~\ref{subsec: parameter variation}, we systematically vary the most relevant hyperparameters and report which changes have been beneficial to derive equilibrium outcomes in bidding problems. It is our hope to contribute thereby to building up similar domain specific recommendations for deep learning in energy systems analysis as is currently available for image recognition.

\subsection{Main Benchmark Scenario: Capacity Constrained Bertrand Competition}
We briefly recapitulate the structure of our benchmark scenario detailed in Section ~\ref{Sec:BenchmarkingScenario}. Our setting consists of a duopoly competing in offer prices. The competition takes place through an uniform price auction mechanism. The auction mechanism has the objective to maximize welfare given the firm's offer prices and returns a single (uniform) market-clearing price and market-clearing quantities for both players. We assume that both players have a fixed capacity and submit offer prices simultaneously to the auction mechanism. The static game has a continuum of equlibria that can be described analytically. In the learning context however, the one-shot game will be iterated many times and agents may learn how to optimally navigate in such a setting. 

We assume symmetry, i.e., both players are identical in terms of cost and maximum capacity (50) they control. Both players are assumed to always offer their full capacity and may submit offer prices in a continuous interval that we have normalized between $[-100,100]$. We assume a production cost per unit of $20$ and constant inelastic demand of $70$.\footnote{Most of these assumptions can be relaxed easily to model more realistic cases. However, the main purpose of our paper is to test specifications of the algorithms that is why we have chosen an analytically tractable bechmark scenario.}

We paramatrized our benchmark model such that demand is strictly larger than each firm's capacity. Hence, both firms will have positive market-clearing quantities. However, due to the uniform pricing auction setting, the marginal firm, i.e., the firm that will set the price, will produce less than the infra-marginal firm. Hence, the trade-off in this game is between setting the price but selling less, and being infra-marginal (not setting the price) but selling full capacity. Clearly, the latter strategy will dominate in terms of per-unit profit given one player is setting a high price and that both players are symmetric. We find that equilibrium strategies typically converge to states where one player is offering at the price cap and the other one far below---a behavior consistent with the analytical equlibrium description.

The advantage of deploying DDPG to the problem rather than $Q$-learning is that the former is able to handle continuous action spaces rather than only discrete actions spaces. Nonetheless, with the right parameterization, we achieve a reliable convergence of the DDPG algorithm despite the continuous action space. For instance, in Figure~\ref{fig:LearningExample} we depict the learning progression of a well-tuned DDPG run. 

DDPG also remains effective in finding equilibrium outcomes. Figures~\ref{Subfig:3norm_nomemory} and \ref{Subfig:3norm_memory} depict the share of 100 test runs that have terminated in a Nash Equilibrium state. The blue, green and red line describe runs with differing normalization schemes (see Section~\ref{Subsec:Normalization}), the height of these lines intersection with the orange Nash Equilibria border is the ratio of runs that converged to equilibrium within 15,000 episodes.\footnote{An episode corresponds to a single round of submitting bids and subsequent market clearing.}
There, we see that the best performing method (i.e., layer normalization, red line) achieves convergence rates of $100\%$ and $~99\%$ in the small and large state-space case respectively.
Therefore, it is evident from Figures~\ref{Subfig:3norm_nomemory} and \ref{Subfig:3norm_memory} that it is possible to attain reliable convergence in continuous state-action spaces with DDPG.

\begin{figure}[h]
     \centering  
     \includegraphics[width=\textwidth]{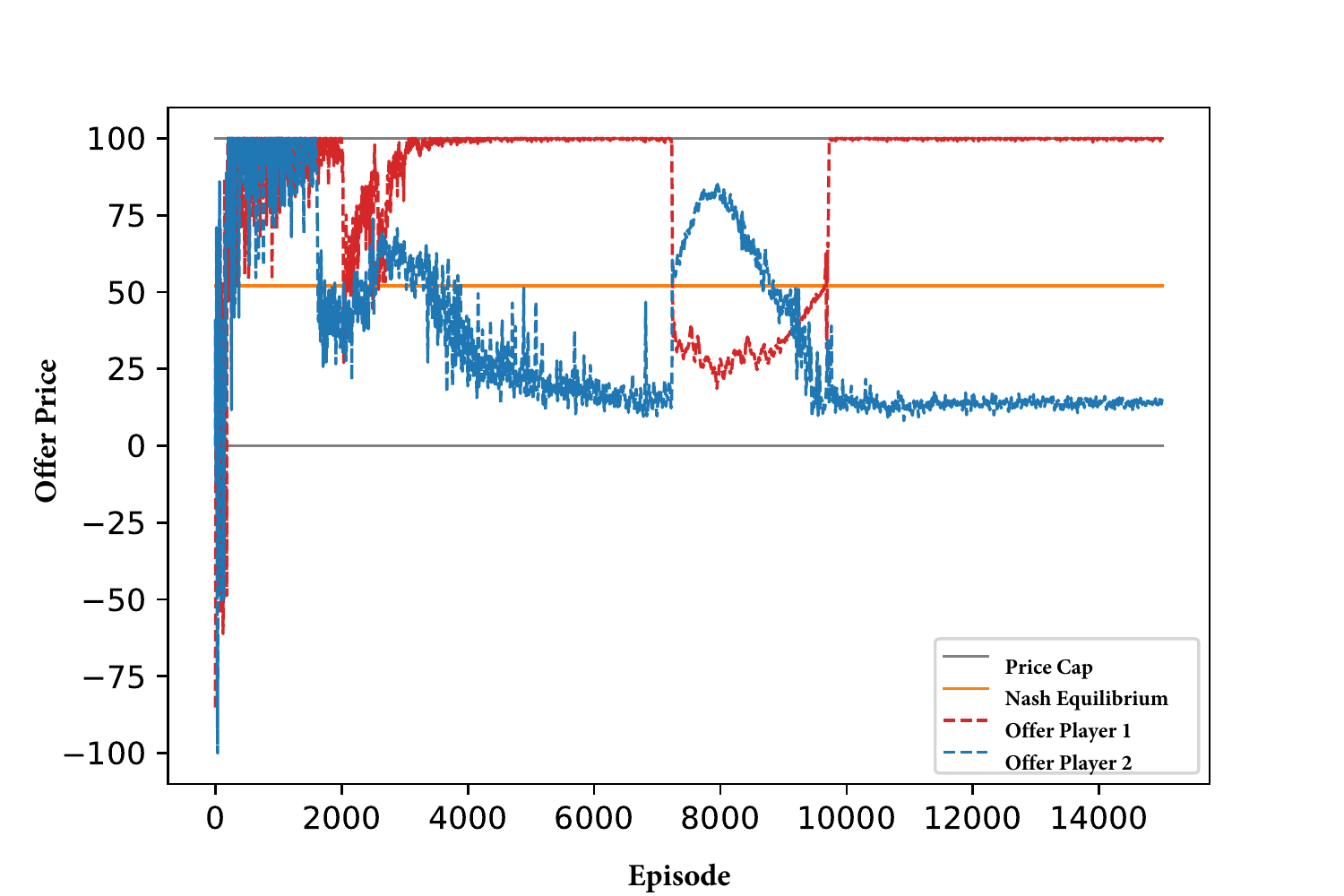}
          \caption{Exemplary learning progression}
        \floatfoot{\footnotesize \textit{Notes:} 
        We plot the development of offer prices in a duopoly.
        Each point corresponds to a players bid price in a given episode. The plot shows different phases of the learning progress throughout an exemplary DDPG run. We can distinguish:
        a) Episode 0-2000: Random exploration, exploration of action space limits
        b) Episode 2000-4000:
        Competition and exploration of opponents reactions
        c) Episode 4000-7000:
        Initial relaxation to equilibrium with high variance
        d) Episode 7000-9000:
        Attempted exploration of alternative strategies
        e) Episode over 9000:
        Final stabilization to equilibrium with low variance
        }
        \label{fig:LearningExample}
\end{figure}

\subsection{Alternative Benchmark Scenario: Unconstrained Bertrand Competition}

The ratio between demand and the agents capacities is a critical parameter. We have discussed the details in Section \ref{Sec:BenchmarkingScenario}, where we argued that the most interesting case is $\overline{q}<D<2\overline{q}$, as covered in our standard Benchmarking case. Nontheless, $0 < D\leq \overline{q}$ is worth a brief discussion, since it reduces to unconstrained Bertrand competition and hence the most well-known type of Bertrand duopoly. It is well-known that the only equilibrium in this case is the one, where both players bid their marginal costs.
Albeit, being of minor theoretical relevance, reproducing the behavior of Bertrand competition constitutes an important first validation of our algorithms agreement with established theory. Indeed, this seemingly ``easy'' case yields some interesting side results. 

In a slight departure from our standard benchmark case, we have performed a couple of algorithm runs implementing a standard Bertrand duopoloy with $\overline{q}_1=\overline{q}_2=D=50$ instead of our standard choice $\overline{q}_1=\overline{q}_2=50, D=70$. 

In essence, we find that DDPG converges towards marginial costs in the standard Bertrand duopoly as expected. However, there exists a caveat: the speed of convergence. Indeed, it is sometimes possible that DDPG's offers approach the marginal costs at a creeping speed leading to convergence times of roughly 100,000 episodes. 
Well parametrized runs of capacity constrained Bertrand duopolies typically do not exceed convergence times of 20,000 episodes.

Moreover, we were able to identify at least two criteria that control the convergence times of unconstrained duopolies: the size of the memory buffer \& noise decay rate. Indeed, this case makes it particularly apparent that the size of memory buffer needs to be chosen carefully. Both too large and too small buffer sizes negatively impact the convergence of the ``easy'' unconstrained problem.
Similarly, the noise decay needs to be adjusted, this finding is however expected (well-known as exploration-exploitation trade-off) and similar in all RL algorithms. 
We discuss the details in the Section~\ref{subsec:D=CAP}.

Finally, we want to remark that (even extremly) slow convergence to an equilibrium (such as marginal costs in unconstrained Bertrand competition) by no means contradicts any of game theories fundamental predictions. GT is not at all making any assumption \textit{how} or \textit{how fast} players arrive at equilibria. GT solely predicts that eventually players ought to arrive an remain within equilibrium. It may well be the case, that certain equilibria are relatively slow to learn despite having well-known theoretical justifications. 

\subsection{Tuning Parameters}
\label{subsec: parameter description}

DDPG is considered a so-called actor-critic method, thus it uses two neural networks (the actor and the critic network) in parallel with possibly differing design choices. Hence, DDPG requires a large set of parameters to be tuned.
A relevant contribution of our work is to evaluate sensible values and relevant parameters for the problem domain of uniform price auctions.
However, due to the many possible interactions between parameters a selection has to be made, which parameters are investigated in detail and which are held constant. 
We have deemed the following parameters to be most relevant and thus studied their significance in Section~\ref{subsec: parameter variation} through a variational analysis:
the \textit{actor learning rate},
the \textit{critics learning rate},
inclusion of \textit{last rounds actions in the state space},
several \textit{action normalization schemes}. 
While we held \textit{replay buffer size} fixed in the main benchmarking scenario (capacity constrained Bertrand competition), we performed a variational analysis of replay buffer size for the unconstrained Bertrand duopoly. Moreover, we have found significant differences between constrained and unconstrained Bertrand duopolies.
Indeed, in the unconstrained case we have found the \textit{replay buffer size} to be the decisive parameter alongside the \textit{noise decay rate}, although these have no been relevant for convergence in the constrained case.
%and the \textit{discount factor} $\gamma$.

Apart from these variational analysis, the following parameters, chosen almost entirely in accordance with \cite{lillicrap2015continuous}, have been held fixed: the choice of an almost empty \textit{state-space} except for the inclusion or exclusion of last rounds actions, the \textit{optimizer},  \textit{actor network design} including depth and hidden layer, \textit{critic network  design} including depth and hidden layer, \textit{soft-update rate}$\tau$, the \textit{noise}, and the use and size of a \textit{replay buffer}. We have summarized our choice of parameters in Table~\ref{Fig:TableParameters} and depicted the overall structure of our neural network architecture and the flow of reward-action-signals in Figure~\ref{fig:DDPGflow}.

\begin{figure}[htbp]
     \centering
     \caption{Depiction of the employed Multi-Agent Neural Network Architecture}
     \includegraphics[width=\textwidth,height=0.45\textheight,keepaspectratio]{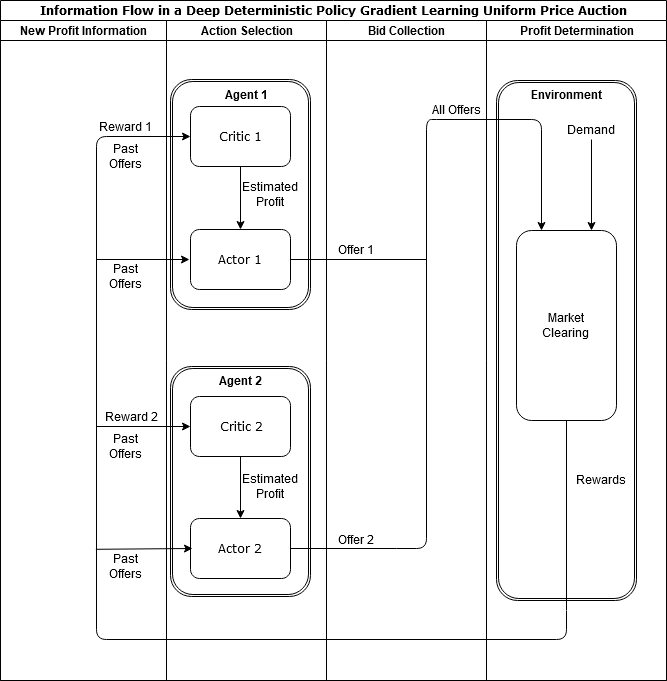}
    \floatfoot{\footnotesize \textit{Notes:} We depict the interaction between agents and environment in our uniform price auction scenario.
    Each Agent is equipped with two neural networks: the Critic and the Actor. In each rounds the critic processes the actual rewards and produces future reward estimates, that are in turn fed to the Actor that uses them to determine a price bid. In case of explicit state space memory, all neural networks receive last rounds bids as state-space information. All agents submit their bids to the environment. 
    The environment combines all submitted bids with demand information and a market clearing subroutine. The market clearing determines the market price, sold quantities and revenue. Revenue is turned into profit within the environment by subtracting costs. Here, profit translates to the reinforcment learning notion of reward.
    Finally, rewards and last turns bid are relayed back to the agents, who can use it to determine next rounds bids. This completes a single learning cycle. }
    \label{fig:DDPGflow}
\end{figure}

%Parameters Table
\begin{table}[h]
\centering
\caption{Model and Hyper Parameters}
\label{tab:1}
\vspace{0.1cm}
\begin{tabular}{lcllc}
\cline{1-2} \cline{4-5}
\multicolumn{2}{|c|}{\textbf{Model Environment Parameters}}              & \multicolumn{1}{l|}{} & \multicolumn{2}{c|}{\textbf{General Hyper   Parameters}}                  \\ \cline{1-2} \cline{4-5} 
\multicolumn{1}{|l|}{Capacities}            & \multicolumn{1}{c|}{50}    & \multicolumn{1}{l|}{} & \multicolumn{1}{l|}{Soft-update rate $\tau$}                  & \multicolumn{1}{c|}{1.00E-03} \\
\multicolumn{1}{|l|}{Marginal Costs}        & \multicolumn{1}{c|}{20}    & \multicolumn{1}{l|}{} & \multicolumn{1}{l|}{Discount rate $\gamma$}                & \multicolumn{1}{c|}{0.99}     \\
\multicolumn{1}{|l|}{Price Cap}             & \multicolumn{1}{c|}{1}   & \multicolumn{1}{l|}{} & \multicolumn{1}{l|}{Max Memory Size}      & \multicolumn{1}{c|}{50,000}    \\
\multicolumn{1}{|l|}{Demand}                & \multicolumn{1}{c|}{70}    & \multicolumn{1}{l|}{} & \multicolumn{1}{l|}{Batch Size}           & \multicolumn{1}{c|}{128}      \\ \cline{1-2}
                                            &                            & \multicolumn{1}{l|}{} & \multicolumn{1}{l|}{Total Runs}           & \multicolumn{1}{c|}{100}      \\ \cline{1-2}
\multicolumn{2}{|c|}{\textbf{Noise Hyper   Parameters}}                  & \multicolumn{1}{l|}{} & \multicolumn{1}{l|}{Episodes per Run}     & \multicolumn{1}{c|}{15,000}    \\ \cline{1-2} \cline{4-5} 
\multicolumn{1}{|l|}{Mean $\mu$}                    & \multicolumn{1}{c|}{0}     &                       &                                           & \multicolumn{1}{l}{}          \\ \cline{4-5} 
\multicolumn{1}{|l|}{Variance $\sigma$}                 & \multicolumn{1}{c|}{0.1}   & \multicolumn{1}{l|}{} & \multicolumn{2}{c|}{\textbf{Neural Network Hyper Parameters}}             \\ \cline{4-5} 
\multicolumn{1}{|l|}{Regulation Coefficient} & \multicolumn{1}{c|}{10}    & \multicolumn{1}{l|}{} & \multicolumn{1}{l|}{Learning Rate Actor}  & \multicolumn{1}{c|}{1.00E-04} \\
\multicolumn{1}{|l|}{Decay Rate}            & \multicolumn{1}{c|}{0.001} & \multicolumn{1}{l|}{} & \multicolumn{1}{l|}{Learning Rate Critic} & \multicolumn{1}{c|}{1.00E-03} \\ \cline{1-2} \cline{4-5} 
\end{tabular}
\label{Fig:TableParameters}
\end{table}

\subsubsection{State Space}
\label{Subsec: State Space Model}

The most fundamental modelling choice in any RL algorithm is how the state-action space is modelled. In our setting the action space is the set of possible offer prices that can be submitted to the market-clearing algorithm. In terms of the state space it is less clear what should be encoded in order to capture the dynamics of an environment. We believe that including richer state spaces is a very promising line of future research, but have decided to keep state information limited to keep the analysis tractable. Consequently, the state was chosen to be almost empty with one exception: the offer prices submitted to the market-clearing of the previous period. 
We emphasize that all environment parameters (i.e., capacity, marginal costs, price cap and demand) have been kept constant, but have not been included explicitly within the state-space. 
Hence, all information regarding the environment parameters is learnt emergently during the algorithm without input from our side.
Although, DDPG includes a parallel memory mechanism (the replay buffer), we opted to explicitly include this information in the state space. Indeed, we find that such an explicit representation of memory impacts the resulting behavior.

\subsubsection{Optimization Routine}
Deep learning algorithms rely on solving numerous successive optimization problems when updating neural networks, hence the choice of solver is relevant for the algorithms behavior.

We have employed the ADAM solver \citep{adam}. ADAM is not only used in \cite{lillicrap2015continuous} but can currently be considered as the state-of-the-art implementation of stochastic gradient descent that is employed throughout the deep learning community. ADAM's popularity comes from the fact that it is an adaptive solver which adjusts initial learning rates and hence promises to work robustly within a range of learning rates. Nonetheless, we find significant influence of learning rates on the convergence of the algorithm.

Moreover, we comment on the common observation that ADAM's performance is not scale independent. In our opinion this is an undesirable side-effect.
However, this effect seems not be limited to the ADAM solver but is consitently reported throughout the deep learning community \cite{deeplearningscaledependance}.
This leads to the counter-intuitive fact that problem hardness is affected by the absolute values of the encountered rewards and actions. Nonetheless, we acknowledge ADAM's prevalence in the deep learning community. Most deep learning practitioners therefore employ a downscaling of the relevant parameters prior to learning and a subsequent rescaling after learning to match the environments scaling again. We follow this common practice and find considerable impact of rescaling in general and the employed scaling method in specific.

\subsubsection{Actor and Critic Network Design}
\label{subsec:NNdesign}

The neural network of the \textit{Actor} consists of four layers, where the first layer is the input layer corresponding to the size of the state, the second and third are rectified linear (ReLU) layers. The hidden layer has a size of 400 input nodes and 300 output nodes. An actor network needs to output an action. The final tangens hyperbolicus layer reduces to output to one node, since each agent can only place one offer per episode. The choice of tangens hyperbolicus activation function may seem exotic at first glance, however it is motivated by mapping any desired interval to the $[-1,1]$ interval. This renormalizes the action space to a value range, where the adam optimizer operates reliable. Therefore, the actor network's output is a normalized action and needs to be renormalized to make sense within the context of the environment.

The \textit{Critic} neural network consists of three layers, where the size of the first layer is the size of the $state$, the second ReLU layer is also 400 plus the size of the $action$ vector (i.e., one), the third layer has also a size of 300 and the last layer equating to size of the $action$ vector, since a $Q$-value for every value within the $action$ vector is needed. Therefore, the critic has a certain asymmetry in design. State information traverses a one layer deeper network than action information. Including the $action$ only from the second layer is suggested by the authors of the DDPG algorithm \citep[see][]{lillicrap2015continuous}.

\begin{figure}[h]
    \centering
     \begin{subfigure}[b]{0.45\textwidth}
         \centering
         \includegraphics[width=\textwidth]{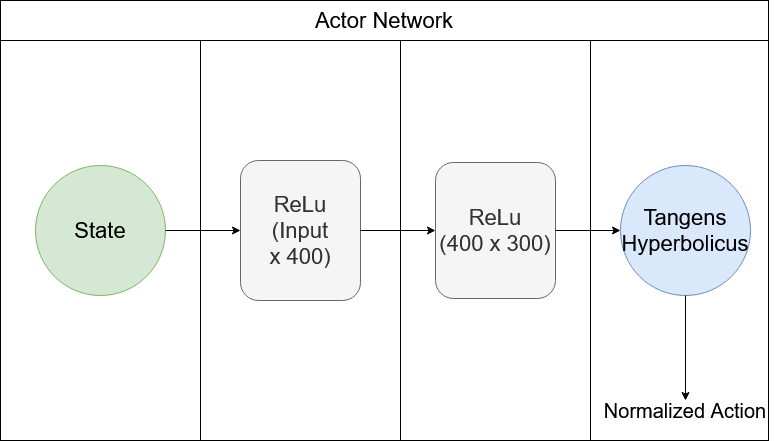}
         \caption{ }
         \label{fig:ActorNN}
     \end{subfigure}
     \hfill
     \begin{subfigure}[b]{0.45\textwidth}
         \centering
         \includegraphics[width=\textwidth]{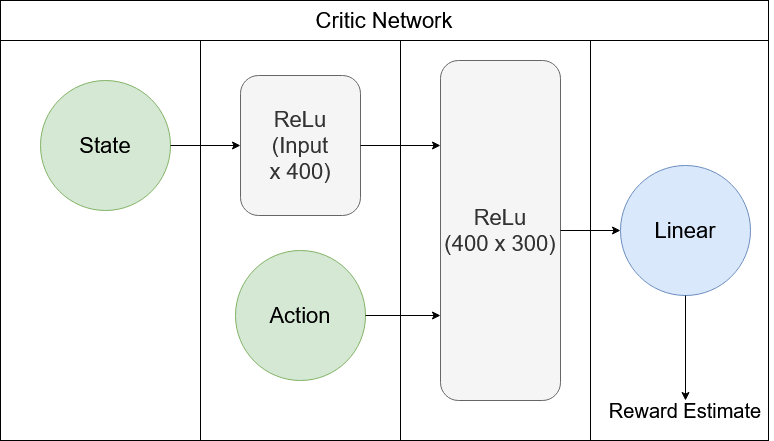}
         \caption{ }
         \label{fig:CriticNN}
     \end{subfigure}
    \floatfoot{\footnotesize \textit{Notes:} 
    Detailed representation of Panel (a): Actor Networks and Panel (b): Critic Networks. All hidden layers use Rectified Linear activation functions.
    The actor network uses one layer to process state information, subsequently performs one coarse graining step and uses a tangens hyperbolicus layer to normalize outputs between [-1,1]. Critic architecture is similar, although action input is directly fed into layer 2. The output layer is linear.}
        \caption{Schematic Representation of Neural Network Architecture}
    \label{fig:NN_architecture}
\end{figure}

\subsubsection{Noise}
\cite{lillicrap2015continuous} originally used an Ohrensteil-Uhlenbeck noise, but recent researches showed that a normally distributed Gaussain noise performes just as well \citep{TD3}. Hence, we have opted for Gaussian noise, since it can be run with fewer hyperparameters thus facilitating isolation of relevant parameters.

We decided to use normal distributed Gaussian noise with a mean $\mu = 0$ and and a variance $\sigma = 0.1$. Additionally we apply a regulation coefficient to move the mean, to enlarge the starting noise to further ensure that the whole $action$ space is explored adequately. To make sure that the algorithm converges with time, there is also an decay rate applied to the noise. The decay ensures that the noise gets smaller with each episode, until a defined minimum noise setting is reached. 

\subsubsection{Replay Buffer}
Another feature of DDPG is the replay buffer. A replay buffer stores a (typically large) number of past rounds. In our case up to 50,000 rounds are stored in the memory buffer. Initially, the buffer is empty and every action is stored within the buffer. The memory buffer is managed according to the first-in-first-out principle i.e., once 50,000 memorized actions are stored, every new action is saved in the buffer, at the expense of deleting the oldest memorized action to adhere to the memory limit.

In the standard benchmark scenario the replay buffer was a fixed parameter with size 50,000.
In this case the buffer is chosen so large, that memory constraints should not influence the algorithms outcome.
It is possible that the large buffer size may have slowed convergence, but due to us being able to attain convergence with proper learning rate choices we opted not to vary the buffer size.

Neural networks do not contain an exact description of past rounds, but are adapted iteratively by minimizing their predictions deviation from a set of target states in each update that learns new incoming information. This allows neural networks to remain small, when scaling up the size of the state-action space, with only minor prediction losses. However, neural networks of DDPG should incorporate both new information, but also be able to retain past knowledge. For this reason, every update of a neural network randomly draws a number of old rounds from the memory buffer. The replay buffer aims to statistically represent the past. Then the network is adapted to attain minimal loss on the new rounds and the rounds drawn from the memory buffer.

The replay buffer is also a hallmark feature of more modern neural network methods, that typically employ them. 
In contrast, classical reinforcemnt learning such as $Q$-learning usually have no directly corresponding type of memory.
In $Q$-learning, memory can be included into the state-space. 
However, the replay buffer is in fact a second type of memory, alongside a possible inclusion of for instance last rounds action into the state-space. Experiments lacking an explicit memory representation in terms of state-space cannot be said to be truly memory-less. Informally, the replay buffer may be seen as a type of approximate ``long-term'' memory, while the state-space memory may be closer to an exact ``short-term'' memory. An indepth analysis of the exact interrelation between both memory processes is left to future research.

\subsection{Variational Analysis of Learning Parameters}

This section covers the impact of parametrization on DDPG's performance.
We will see that correct parametrization is vital to the performance of DDPG, hence we believe that recommendations on best practices in domain specific parameter choices will be a significant contribution.
In order to make an informed parameter choice, we vary the following parameters: \textit{learning rate actor, learning rate critic, normalization methods} and \textit{state-space memory}. 
We analyze the the variations effect on the frequency of found Nash Equilibria, which is our main figure of merit.
Overall, we vary the two learning rates against each other and the normalization methods against each other. All these are once performed with and without explicit state-space memory.

\subsubsection{Learning Rates}
\label{subsec:LR_parameters}

%DDPG employs a so-called Actor-Critic design. Hence, actions and states qualities are evaluated by two differing neural networks: the Actor network and the Critic network respectively. Aside from the differing network designs discussed in Section \ref{subsec:NNdesign}, the most relevant design choice of a neural network is the learning rate. 
A priori it is not clear how to set the learning rates in both of DDPG's neural networks nor whether they should be the same. Therefore, we vary the actor and critic learning rates independently and around the recommended values in \cite{lillicrap2015continuous}. More specifially, we vary both learning rates in the interval $[1e^{-2},1e^{-5}]$ with four equidistantly chosen learning rate steps per order of magnitude. Consequently, each interval is traversed in a resolution of 13 settings. This means that we have explored 169 learning rate combinations. Each learning rate setting, has been run 10 times independently with identical parameter settings. We have performed statistics on the number of runs successfully converging to equilibrium strategies after 15,000 iterated epsiodes and summarized the outcomes in Figure~\ref{Fig:LR_heatmap}.

\begin{figure}[th]
     \centering
          \begin{subfigure}[b]{0.49\textwidth}
         \centering
     \includegraphics[width=\textwidth]
         {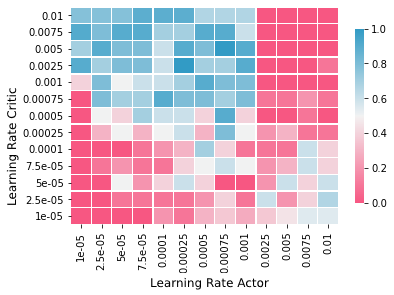}
         \caption{Empty State-space \\ (No Memory)}
     \end{subfigure} 
     \begin{subfigure}[b]{0.49\textwidth}
         \centering
         \includegraphics[width=\textwidth]
         {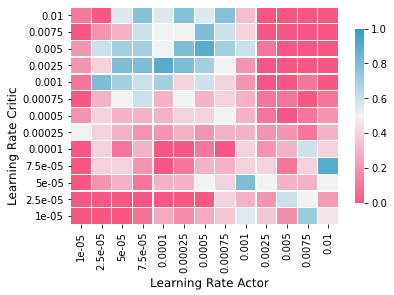}
         \caption{Small State-space \\ (with Memory)}
     \end{subfigure}
\floatfoot{\footnotesize \textit{Notes:}
The colour pixels represent the convergence rate of 10 DDPG runs with a specific combination of learning rates.
Each panel plots parameter choices between 0.001--0.000001 of the actor networks learning rate against the critics learning rate.
Blue colours correspond to high convergence rates, red colours correspond to low rates.
Panel (a) showcases an empty state-space, while Panel (b) is otherwise identical but contains last rounds actions explicitly in the state space. Panel (a) shows overall more blue pixels than Panel (b) indicating that more complex state-space are more sensitive to correct learning rate choice. Nonetheless, both Panel (a) and Panel (b) contain pixels with 100\% convergence rates, implying that with correct parameter choices complete convergence is possible even with the larger state space of b). }
\caption{Rate of convergence to Nash-equilibrium with varying learning rates}
\label{Fig:LR_heatmap}
\end{figure}

The heatmaps in Figure~\ref{Fig:LR_heatmap}, plots the respective learning rate of the actor and critic networks against the percentage of found equilbrium strategies. Color code blue [red] corresponds to high [low] shares of successful conversions to equilibrium strategies. The two heatmaps in Figure~\ref{Fig:LR_heatmap} depict the differences between a memoryless empty state-space (Panel a) and a small state-space that solely stores the actions of the last round (Panel b).

Generally speaking, we find that rates of high convergence are concentrated in the upper left quadrant of the heatmap. This quadrant corresponds to actor learning rates below 0.0025 and critic learning rates above 0.0025. Choosing learning rates of the actor and the critic to be either very small or very large impairs convergence. Furthermore, the diagonal representing balanced choices does yield suboptimal convergence results with best results on the $0.001 \times 0.001$ field. While large values of actor learning rates above $0.0025$ appear consistently detrimental, with a surprisingly sharp border between $0.001$ and $0.0025$, critic learning rates may be chosen more tolerantly.

The general pattern is similar in both heatmaps, however it is evident that the blue region of convergence is significantly smaller in the larger state space case in Figure~\ref{Fig:LR_heatmap}, Panel~(b). This shows a visible impact on convergence despite a relatively modest growth in absolute state space size of two extra-variables. Furthermore, it demonstrates that in our specific context supplying additional information is not necessarily beneficial. Although the smaller region of convergence is a drawback, we point out that this does not imply that the algorithm converges to a lesser degree with more state-information. In fact, this is more indicative of the algorithms sensitivity to parametrization. 
Both panels in Figure~\ref{Fig:LR_heatmap} reach maximal convergence rates well above 90\%.
This means that in both panels there exist parameter choices that lead to almost guaranteed convergence.
However, the number of learning rate combinations that lead to convergence above 90\% is lower in Panel (a) than Panel (b). 
This in turn means, that there are less parameter combinations that lead to reliable convergence in Panel (b), i.e., the case with state-space memory.
%Both heatmaps reach the maximal intensity of 100\% convergence rate. 
Hence, well-calibrated algorithms converge reliable regardless of state-space size, while ill-calibrated algorithms fail to converge. What differs is the number of learning-rate combinations that reach high-convergence rate. Large state-spaces work well with fewer learning-rate combinations, thus they are more sensitive to the choice of learning parameters. Therefore, increasing size of the blue areas in the Figure~\ref{Fig:LR_heatmap} can be interpreted as a measure of increasing robustness. This indicates a relation between state-space size and learning parameter robustness.

We have discussed the upper left quadrant in Figure~\ref{Fig:LR_heatmap} at length. However, there remains the lower right corner, where we too see convergence in some settings.
This corner is starring very low actor learning rates such as $0.01$ combined with very low critic learning rates in the order of $10^{-5}$ yielding modest convergence rates. Furthermore, we find that most of these runs are situated on the extreme edges of the action spaces. This is illustrated in Figure~\ref{Fig:lowbid_heatmap}, where we contrast the rate of convergence with small state space against the bid-distribution of the low-bidding player. We only depict the offer prices of the lower bidding player because in equilibrium the high-offering player is supposed to offer at the price cap. Therefore, in equilibrium only the low-bidding player has freedom to vary his offer prices. In the lower right corner of Figure~\ref{Fig:lowbid_heatmap}, Panel (a) and (b), we see the following pattern: the blue-colored area  of converged runs in Panel (a) mirrors the red pattern of extremely low bids in Panel (b). Hence, large actor learning rates combined with small critic learning rates effectively lead to a form of equilibrium selection that favours extremely low bids. It is not entirely clear whether this is desirable, since it is quite possible that these choices are not really the outcome of a working learning algorithm more an over fitting of the action spaces borders. Therefore, we regard the convergence in the lower right quadrant as a less desirable outcome than the convergence in the upper right corner.

\begin{figure}[th]
     \centering
          \begin{subfigure}[b]{0.49\textwidth}
         \centering
     \includegraphics[width=\textwidth]
         {lr_heatmap_woPA_v2.png}
         \caption{Rate of Convergence
         \\
         (Empty State-space)}
     \end{subfigure} \begin{subfigure}[b]{0.49\textwidth}
         \centering
         \includegraphics[width=\textwidth]
         {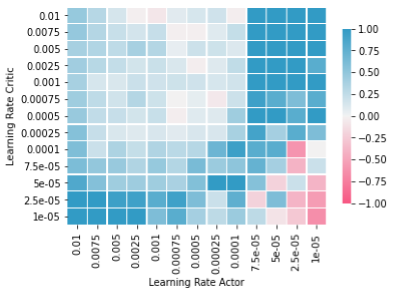}
         \caption{Offer-distribution\\
         (Empty State-space)}
     \end{subfigure}
     \floatfoot{\footnotesize \textit{Notes:}
Panel (a) plots the convergence rate of 10 DDPG runs per pixel depending on actor and critic learning rate. 
Panel (b) plots instead the average of the 10 runs lower bids. Blue corresponds to high and red to low values. We can see a similar curved shape in both Panels. Typically, low convergence rates correspond to both player bidding near the price-cap.
In contrast, while the top-left and bottom-right corner are both convergent in a), while the bid distribution differs significantly in b). All runs were performed in Benchmark Scenario \ref{Sec:BenchmarkingScenario}.}
\caption{Relating Rate of Convergence to Winning Offer-Distribution}
\label{Fig:lowbid_heatmap}
\end{figure}

To summarize, we advise choosing the critic learning rate larger than the actor learning rate. While there is a significant range of valid actor learning rates, typically large learning rates are beneficial for the speed of convergence, if convergence is attained at all. Therefore, we recommend to stay within the $10^{-3}-10^{-4}$ range. We also point out that the inclusion of state variables is a delicate choice. Our results show that it is possible to attain competitive convergence rates with more complex states. However, the relevance of correct parameter choice increases too. Relevant state parameters should be included together with carefully inspection of the learning parameters, but we believe that less-informative or even redundant state variables would have a detrimental effect on the algorithms performance.

\subsubsection{Normalization Schemes}
\label{Subsec:Normalization}
One of the striking methodological differences between classic reinforcement learning ($Q$-learning) and deep learning is the strong prevalence of normalization methods within the latter.

Most machine learning practitioners report significant scale dependencies on the magnitude of the input parameters, when applying deep learning techniques \citep{deeplearningscaledependance}. This surprising effect means that it may be more effective to learn from inputs in an interval $[0,1]$ than $[0,1000]$. These scale dependencies are observed consistently throughout differing problem domains \citep{deepImage,deepSpeech2012}, even apparently scale-invariant problem domains. 
%are affected, this hints at an effect that results directly from the learning techniques deployed in neural networks, rather than the problem formulation itself. 
%Furthermore, classical techniques such as $Q$-learning learning speed are scale-invariant, further strengthening the evidence that deep learning's scale dependency is caused by the methods itself. 
%The effect of scale dependency is well-known and acknowledged in the community, albeit not entirely understood up-to-date. For instance, it is not entirely clear whether scale dependency is a fundamental neural network effect or a possibly avoidable consequence of a design choice such as popular activation functions or commonly employed optimization solvers. Frequent explanations of the effect include the so-called internal covariate shift and numerical stability consideration when computing gradient approximations.
Therefore, down scaling of input data prior to learning and rescaling to the real problem domains scale after learning is standard practice within the deep learning community.

The potential dependence of model outcomes on normalization distinguishes deep learning from classical methods that is why we provide an in-depth analysis of the impact of normalization on learning. We contrast \textit{unnormalized data} with the two most relevant normalization schemes:  \textit{batch normalization} \citep{batchNorm} and  \textit{layer normalization}
\citep{layerNorm}. 
Note, that within DDPG in each step of learning, several neural networks are adapted to a batch of several observations.
Each individual observation is represented by a vector, where each component represents a feature of the state-action space.
Batch normalization curbs the variance between several input vectors (i.e., samples) and layer normalization reduces the variance between the components (i.e., features) of an individual input vector. 
This means that batch normalization normalizes seperately each indvidual features range over several observations.
In contrast, layer normalization normalizes the features within a single observation.
All three normalization approaches are once applied to a problem formulation without state variables (i.e., only actions as inputs) and a formulation that includes last rounds action as state-variables (i.e., ``single round memory'') as elaborated in Section~\ref{Subsec: State Space Model}.

\begin{figure}[ht]
     \centering
     \begin{subfigure}[b]{0.49\textwidth}
         \includegraphics[width=\textwidth]
         {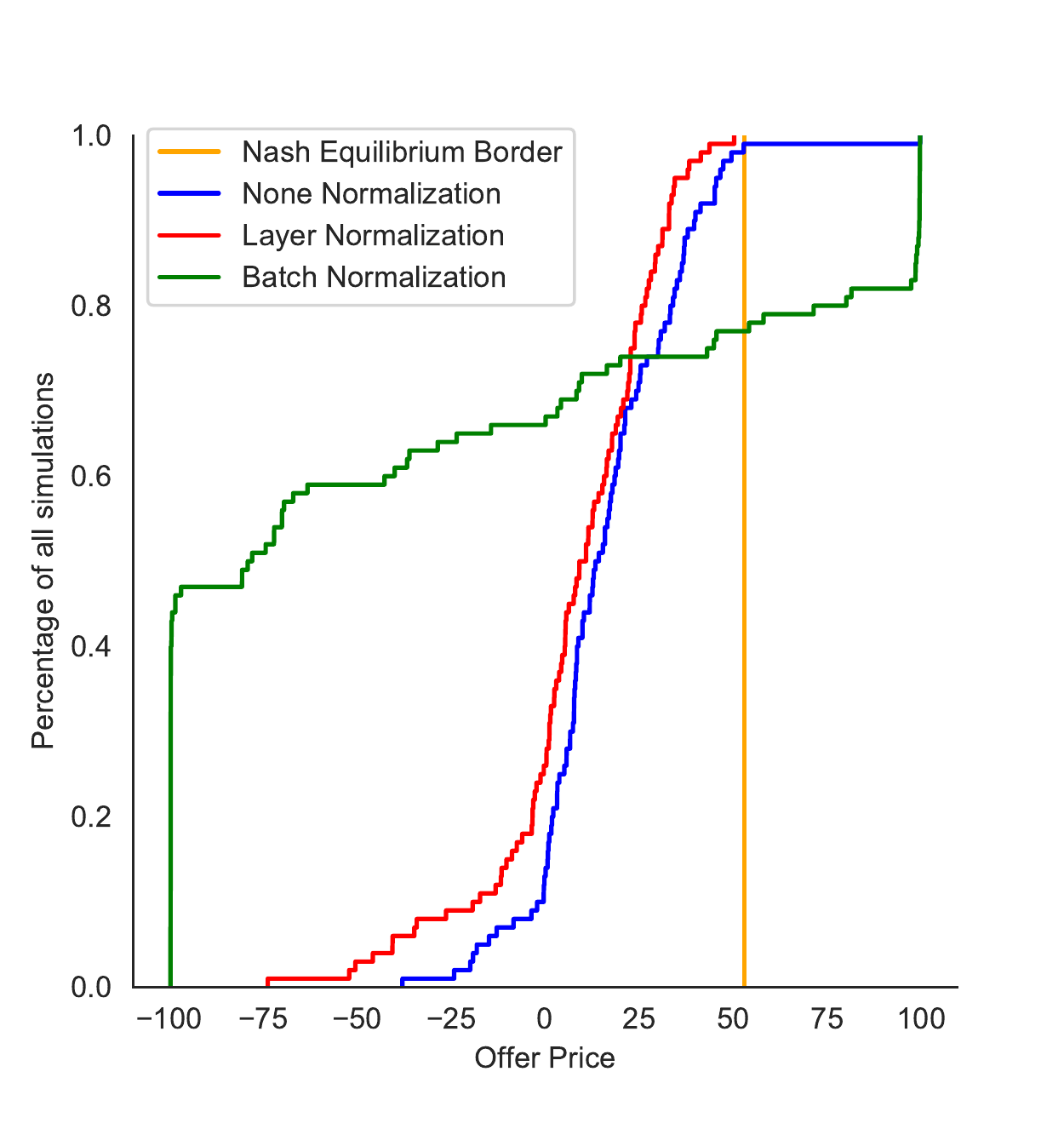}
         \caption{Empty state space (excluding last actions)}
         \label{Subfig:3norm_nomemory}
     \end{subfigure}
     \hfill
     \begin{subfigure}[b]{0.49\textwidth}
         \includegraphics[width=\textwidth]{
         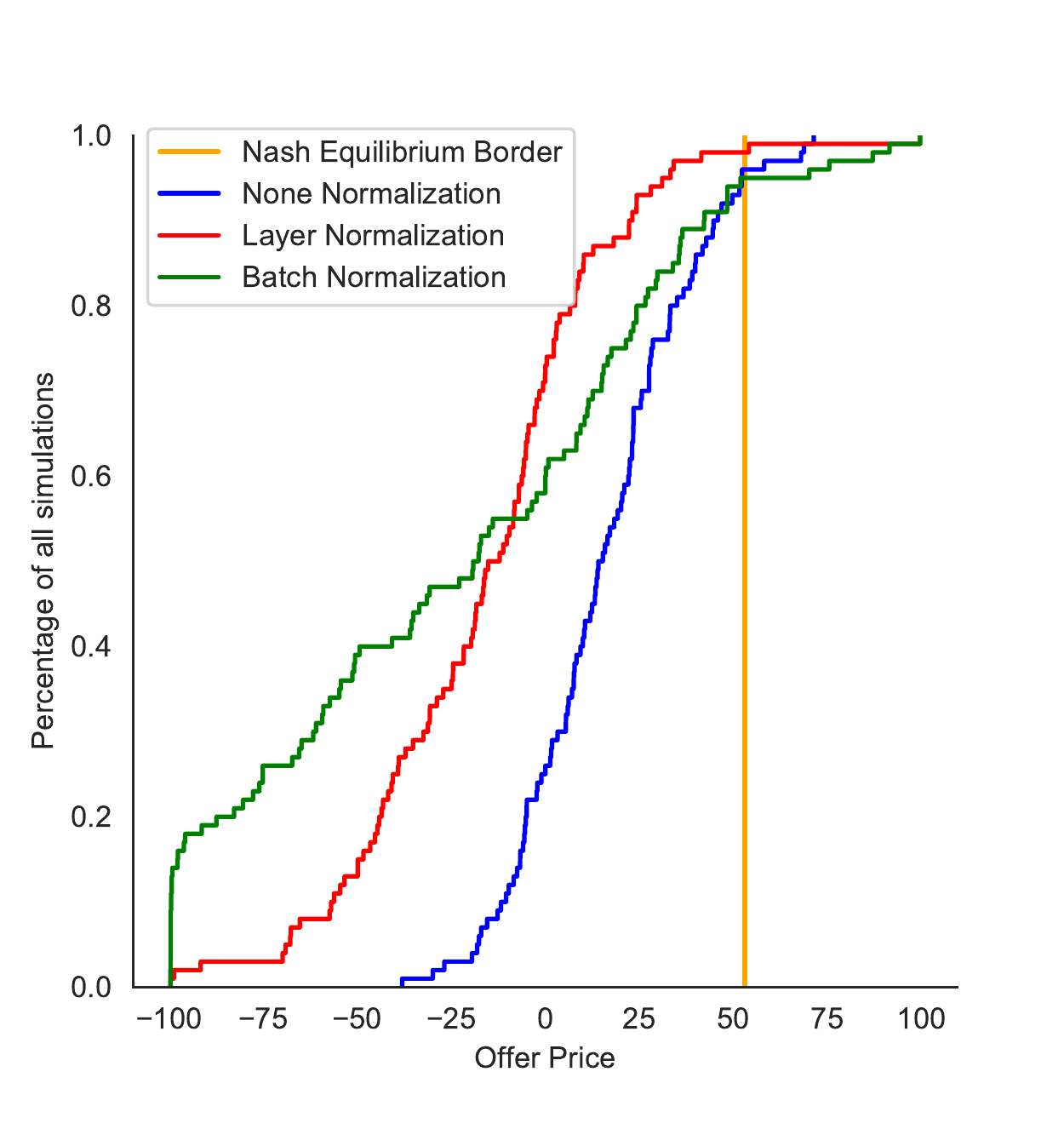}
         \caption{Small state space (including last actions)}
         \label{Subfig:3norm_memory}
     \end{subfigure}
        \floatfoot{\footnotesize \textit{Notes:} We compare 100 unnormalized (blue), batch normalized (green) and layer normalized (orange) DDPG runs in Benchmark Scenario \ref{Sec:BenchmarkingScenario}. Panel (a) depicts runs with empty state space, Panel (b) depicts runs with explicit state-space information about last rounds actions.
        We plot cumulative densities with respect to offer height, thus showcasing the distribution of the low-bidding players offer heights.
        Bids left of the red bar correspond to an equilibrium strategy.
        The graphs plot the percentages of all bids that attained at least the indicated height. Note that the offer-price was rescaled to its natural range after normalization.%Both figures have been run against the Benchmark Scenario \ref{Sec:BenchmarkingScenario}. The left figure (a) has an almost empty state-space, where only the constant demand is included.
        %The right figure (b) posseses the same state-information and additionally includes the last rounds bids explicitly into the state space. This is essentially a single round memory.
        }
        \caption{Distribution of Offer-heights depending on Normalization}
        \label{fig:three graphs}
\end{figure}

Figures \ref{Subfig:3norm_nomemory} and \ref{Subfig:3norm_memory} contrast the impact of different types of normalization schemes and state-space memory information on the number of learnt equilibrium strategies for the low-offering player. The figures of merit depict the percentage of offers within equilibrium. In order to collect statistics all normalization schemes have been iterated 100 times. Remember, that equilibrium outcomes in our setting are characterized as one player offering at the price cap and its opponent offers below the threshold given in \eqref{eq:NE_charakterization_low}. The threshold is determined by the parameters of the game, i.e., price cap, marginal costs, and the ratio of capacity to residual demand. For our standard model parameters (see Table \ref{Fig:TableParameters}), this threshold is calculated as $(\overline{p}-c)\frac{(D-\overline{q})}{\overline{q}} + c = (100 - 20)\frac{70-50}{50} + 20 = 52$. We depicted the threshold as the yellow vertical line in Figures \ref{Subfig:3norm_nomemory} and \ref{Subfig:3norm_memory}. The three remaining lines in each figure represent the cumulative density function, i.e., the distribution of bids derived by differing normalization schemes ordered by their offer height. The intersection point between one of these lines and threshold (yellow) is key, because its height on the y-axis gives the percentage of equilibrium compatible bids. Here, the higher the point of intersection on the y-axis, the better. 

We find that in the memory-less case (Panel a) with small state space, layer normalization (red, approximately 100\%) outperforms ``no normalization'' (blue, approximately 99\%) as well as batch normalization (green, approximately 77\%). In Panel (b), we conduct  the analysis, explicitly including the last round of actions into the state-space. As before, layer normalization (red, approximately $98\%$), outperforms the unnormalized scheme (blue, approximately 93\%) and batch normalization (green, approximately 92\%). Layer normalization has consistently performed best in our setting and is thus recommendable with the caveat that even unnormalized runs performed relatively well. Batch normalization performed worst and cannot be recommended as is. However, we find the strong sensitivity of batch normalization to the increased state space size remarkable. We reemphasize that the increase in size of two additional state variable is relatively small compared to state-spaces size common in the ATARI domain \citep{Atari2015DQN}. Nonetheless, batch normalization improved its performance by at least 20\%, while the other methods got slightly less efficient. This motivates us to conjecture, that batch normalization might proof robust when scaling-up the state-space complexity. Moreover, the relatively good performance of unnormalized runs might be lost in larger state-space, since even a small increase notably had a relevant effect on the unnormalized runs.

Aside from the impacts of normalization on rates of equilibrium convergence, we observe remarkable qualitative differences in the arising offer distributions. For instance, batch normalization exhibits a much flatter offer distribution compared to the unnormalized case at almost identical rates of convergence in the large state-space case (see Figure~\ref{Subfig:3norm_memory}). In the memory-less case, the difference in distribution is even more pronounced, albeit the large difference in convergence makes the distributions overall less comparable. 

In our setting layer normalization performed best. However, the impact of further increases in state-space size seems a promising line of research, since it seems to impact the performance significantly and real-world problems are expected to be equipped with large state-spaces. Furthermore, normalization choices influence the arising offer distributions even at similar rates of convergence. Hence, normalization choices should not be taken solely on a technical basis but consideration is necessary that they also mirror the underlying modelling problem.

% Costs of competition in the 1st 5000
% Differentiate high and low bids, what did the other do?
% Produces Surplus, how much is the cost of equilibration

\label{subsec: parameter variation}
\subsection{Qualitative Analysis}

We have tried so-far to remain as quantitative and explicit as possible when assessing the performance of DDPG to derive equilibrium outcomes in uniform price auctions. We complement this assessment with a brief qualitative discussion. The learning progress throughout a run can vary significantly and many characteristics are hard to distill into a single numerical measure. Nonetheless certain differences in learning behavior are striking when different parametrizations are contrasted: In Figure~\ref{Fig:QualitativeGrid}, we have selected four representative model runs for each of the three normalization schemes. Two runs in each normalization category are memory-less and two have been conducted with single round memory. 
%Note, that there is still considerable variance in each category. 
The chosen runs certainly represent a reoccurring pattern, but nonetheless there are other runs in the same category with distinct appearances. Each category has been run 100 times and we briefly comment on impressions that we had when inspecting these runs.
In order to make exact the individual impressions gained from Figure~\ref{Fig:QualitativeGrid},
we follow up with a statistical evaluation of the offer price development throughout the learning process averaged over all 100 runs with the same normalization and memory choices in Figure~\ref{Fig:Competitiveness}. There, we depict the development of the average competitiveness between the algorithms throughout the learning process.
This is measured for each individual episode by the percentage of runs that undercut each other in a given episode. Overall, competitiveness starts high and eventually decreases, however there are distinct differences in the speed of decline depending on normalization scheme and memory model.

\begin{table}[bthp!]
\begin{tabular}{|cc|c|c|}
\hline
                                                                            & Unnormalized & Layer Normalized & Batch Normalized \\ \hline
\multicolumn{1}{|l|}{\begin{tabular}[c]{@{}l@{}}\rotatebox{90}{No Memory}\end{tabular}} & \includegraphics[width = 1.8in]{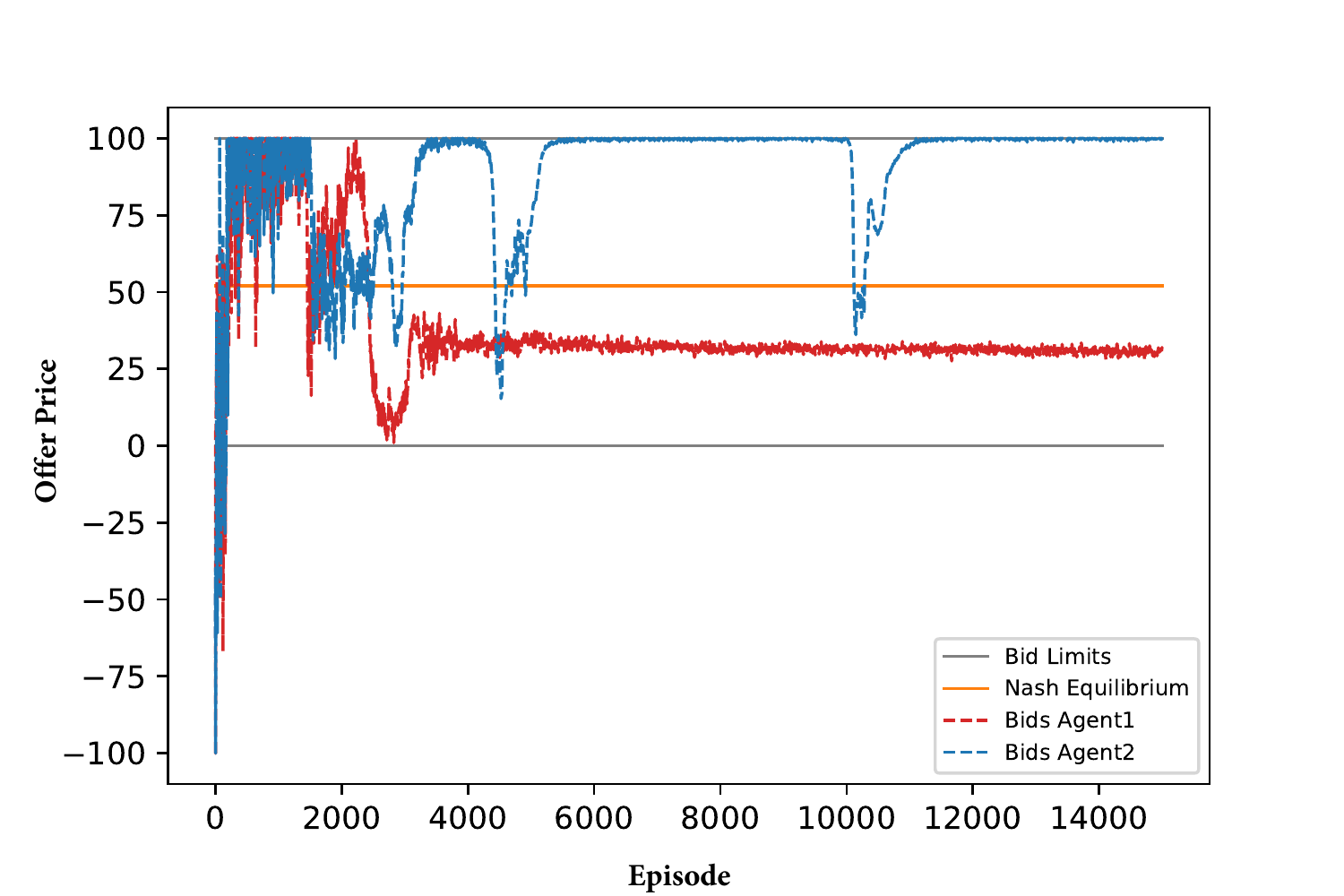}             
                                                                            & \includegraphics[width = 1.8in]{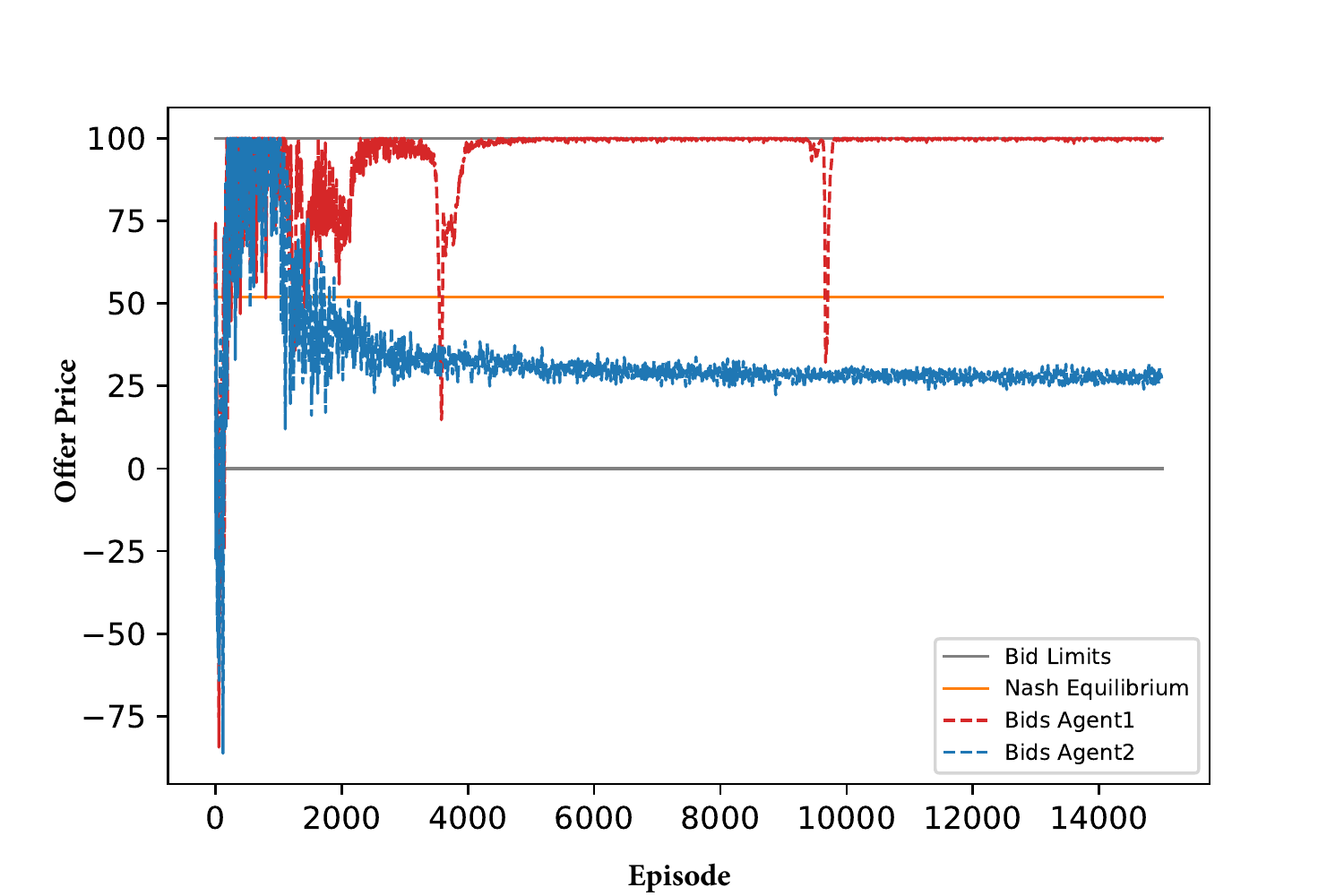}                  
                                                                            & \includegraphics[width = 1.8in]{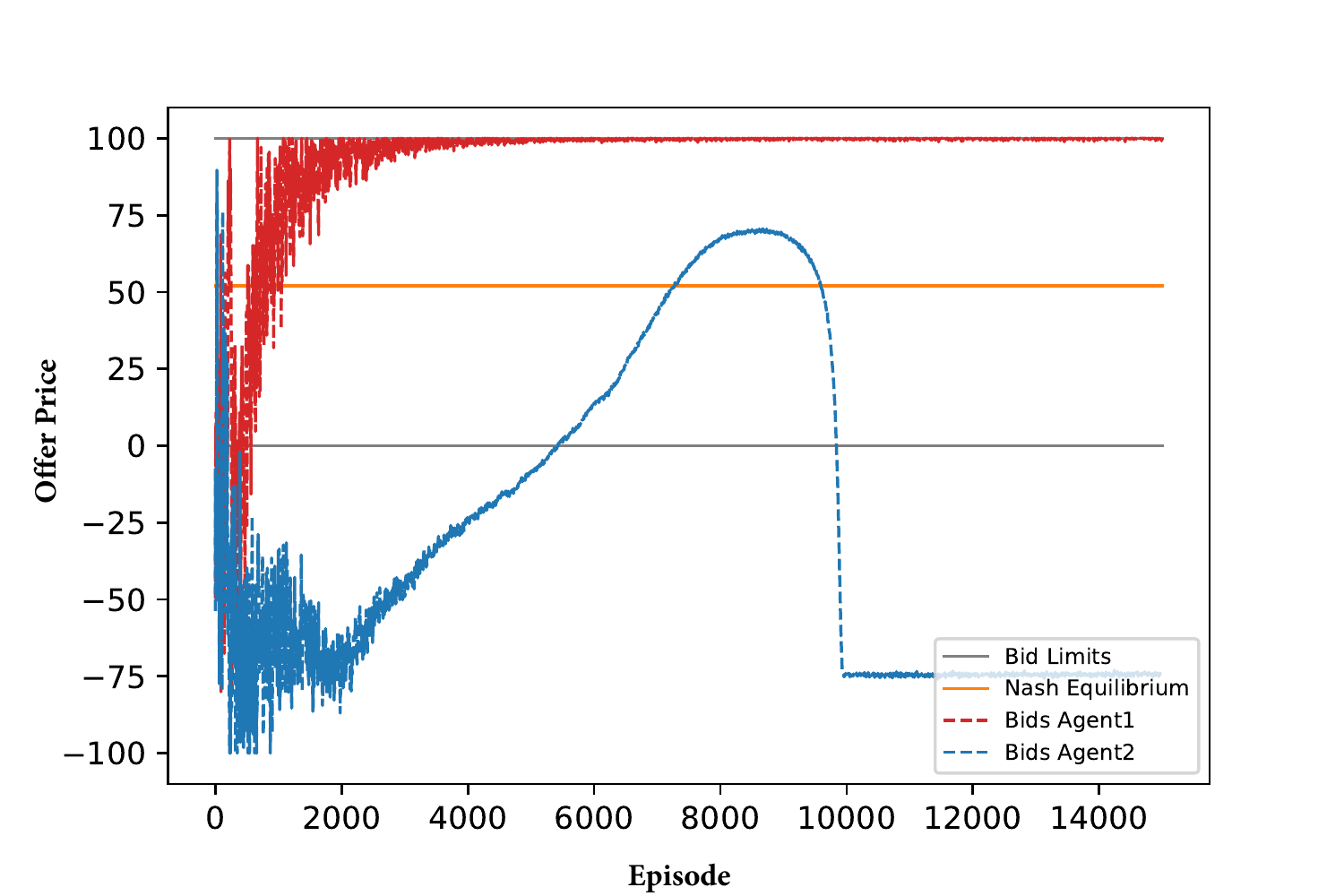}                  \\
\multicolumn{1}{|l|}{}                                                      & \includegraphics[width = 1.8in]{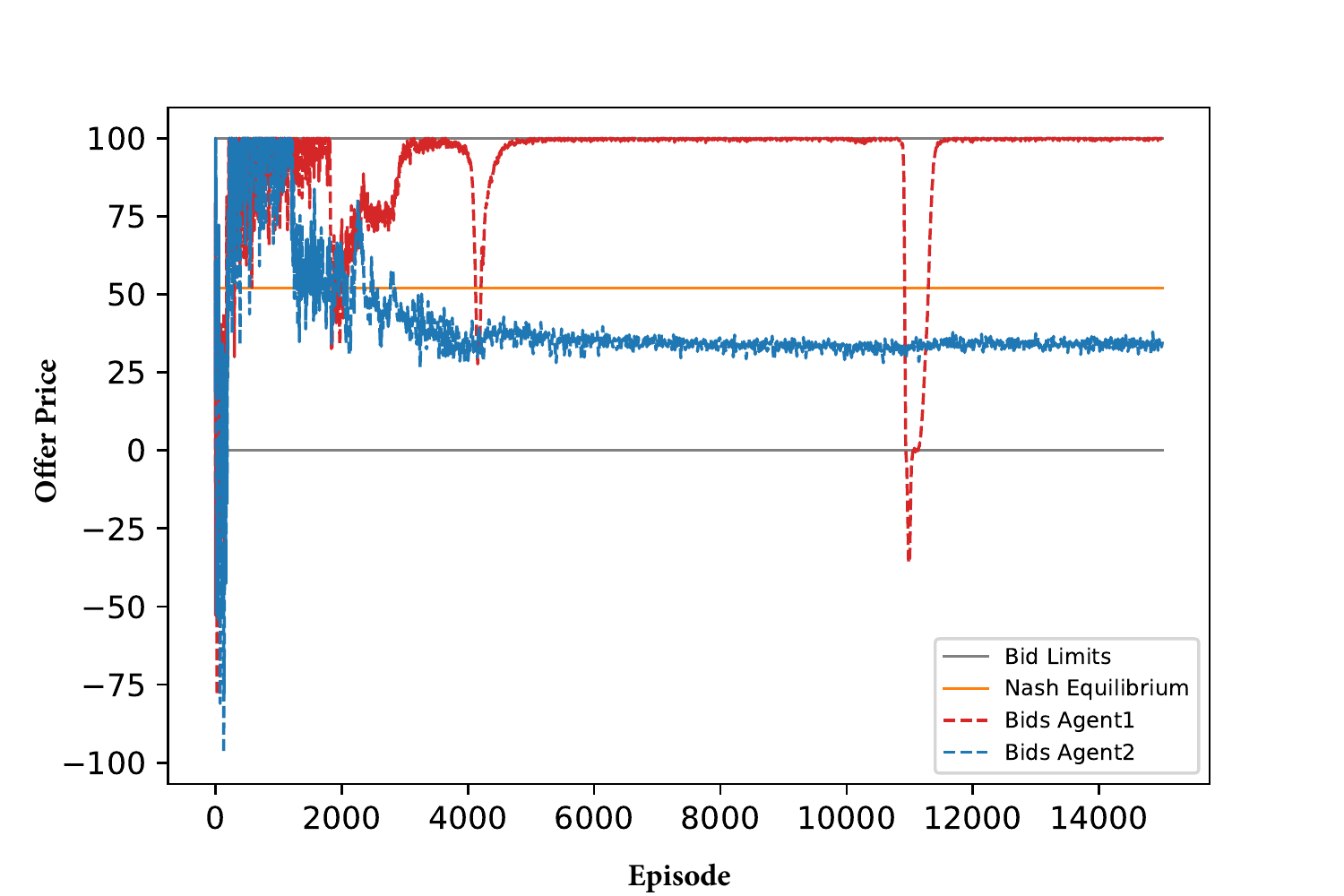}              
                                                                            & \includegraphics[width = 1.8in]{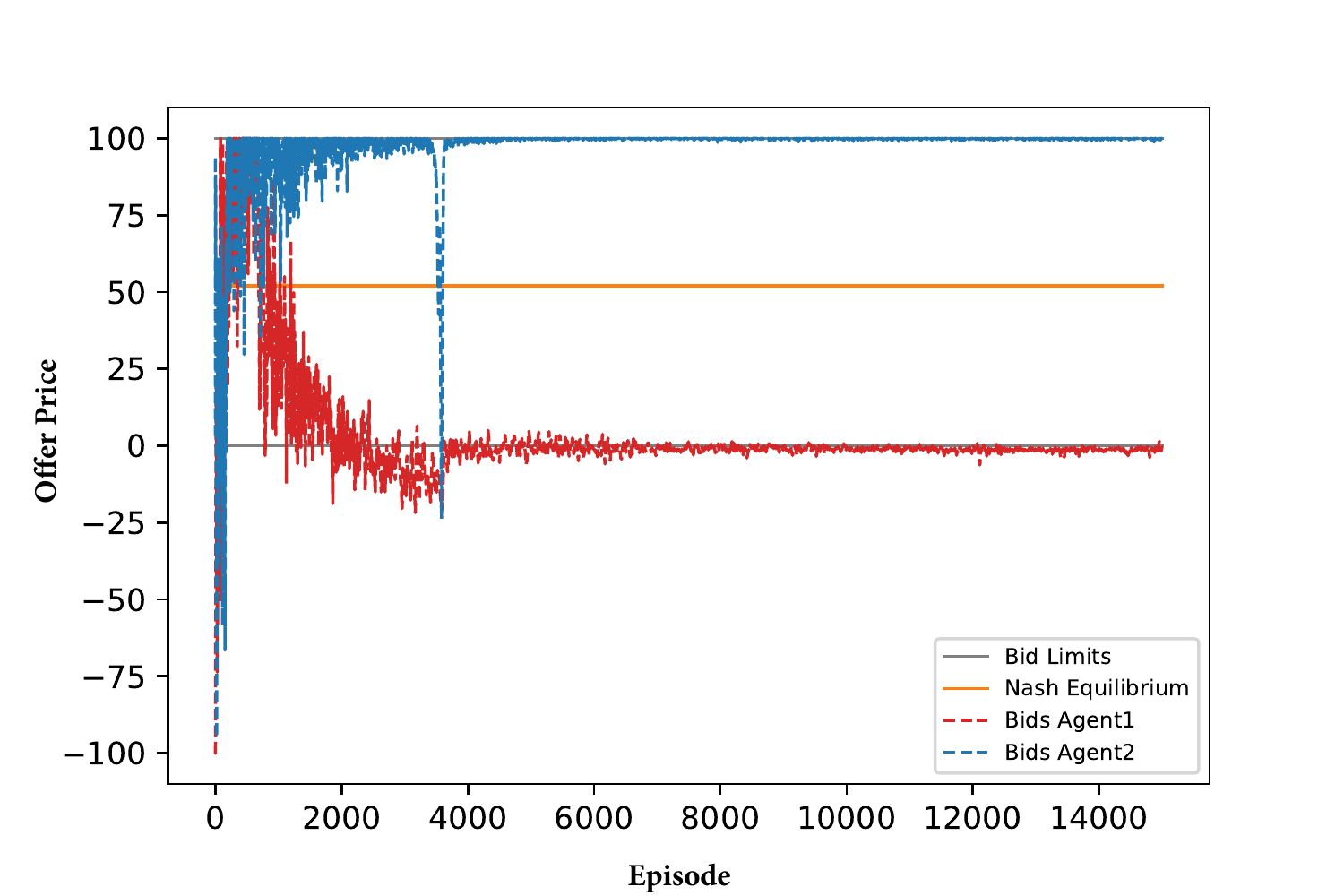}                 
                                                                            & \includegraphics[width = 1.8in]{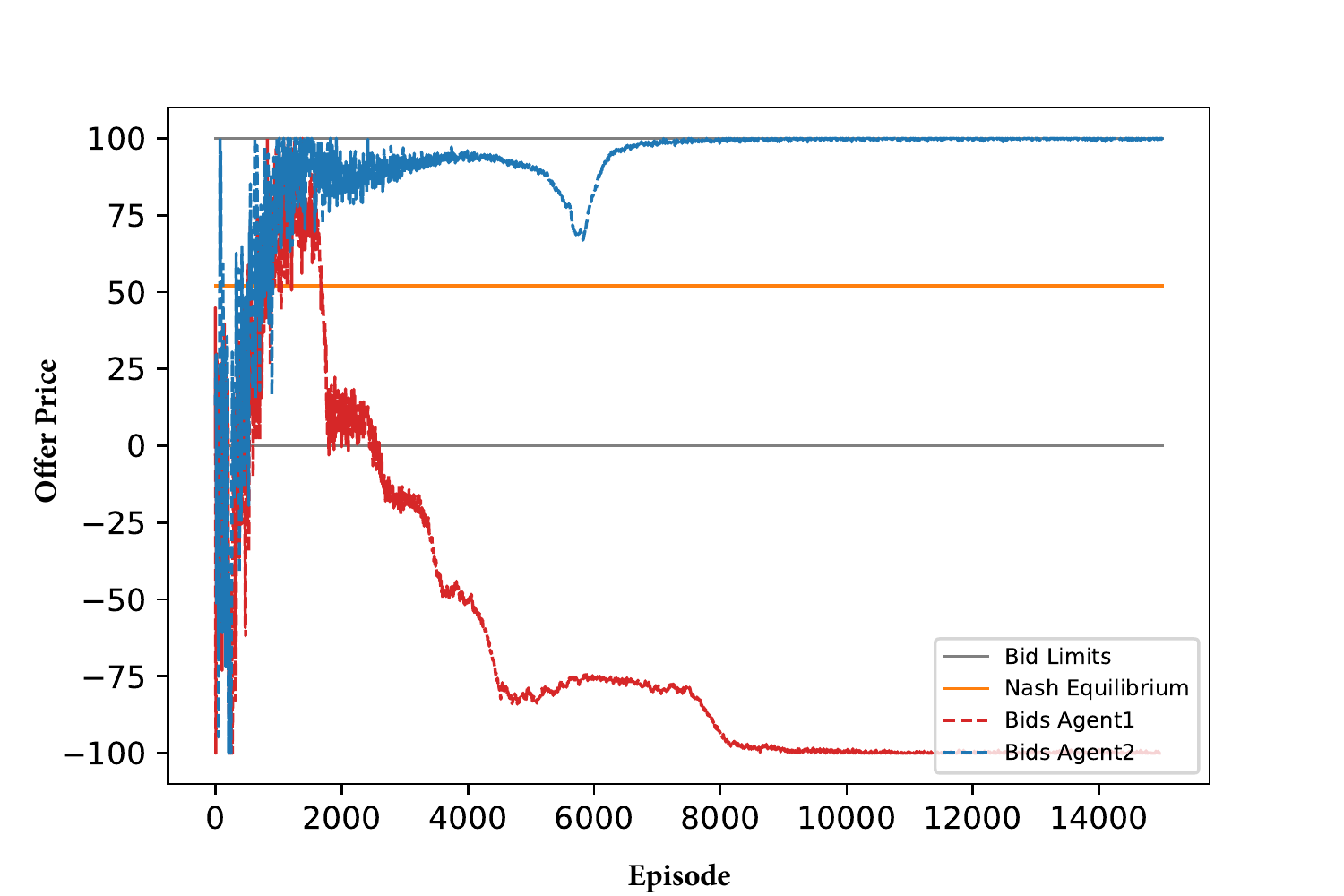}                 \\ \hline
\multicolumn{1}{|l|}{\begin{tabular}[c]{@{}l@{}} \rotatebox{90}{With Memory}\end{tabular}} & \includegraphics[width = 1.8in]{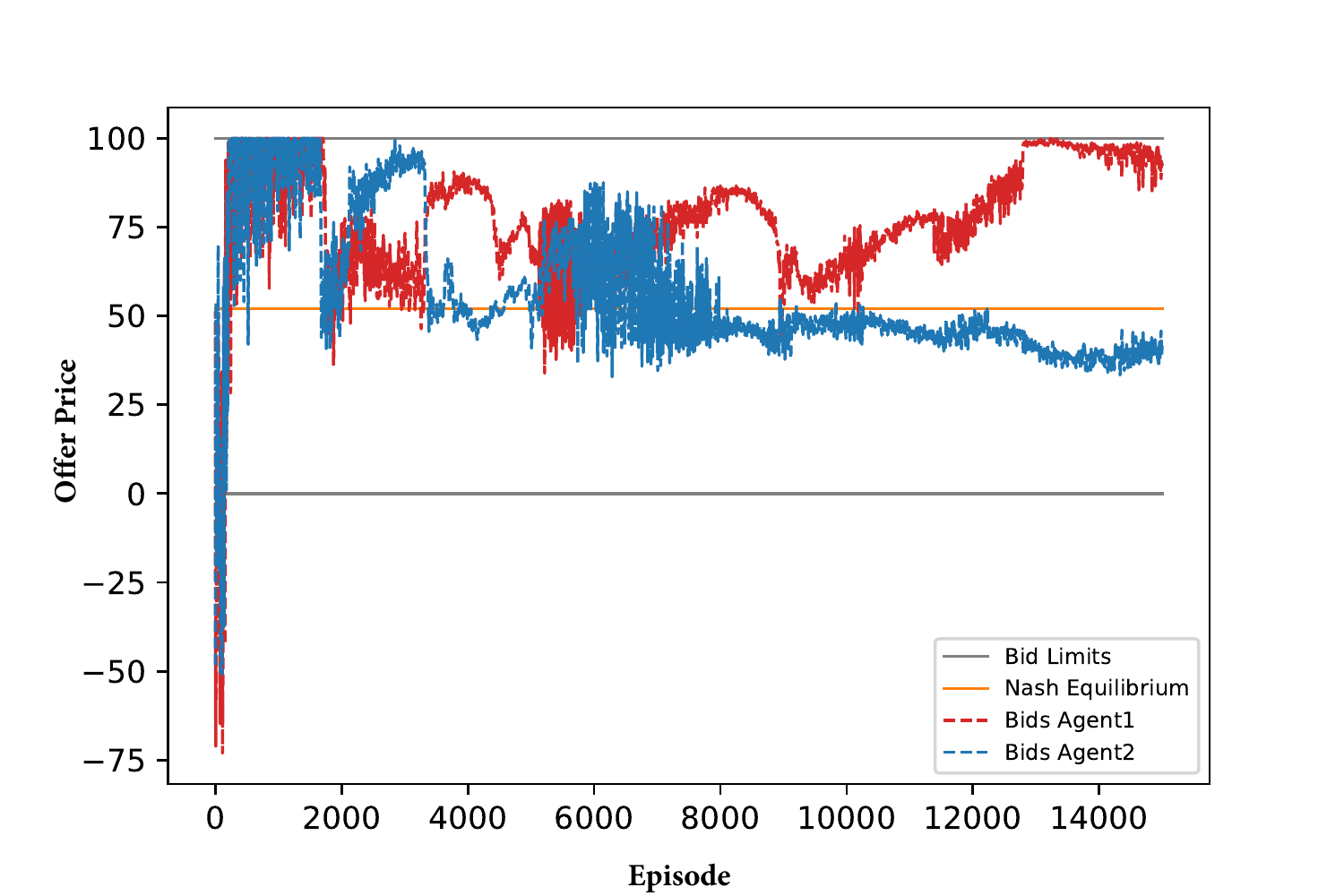}             
                                                                            & \includegraphics[width = 1.8in]{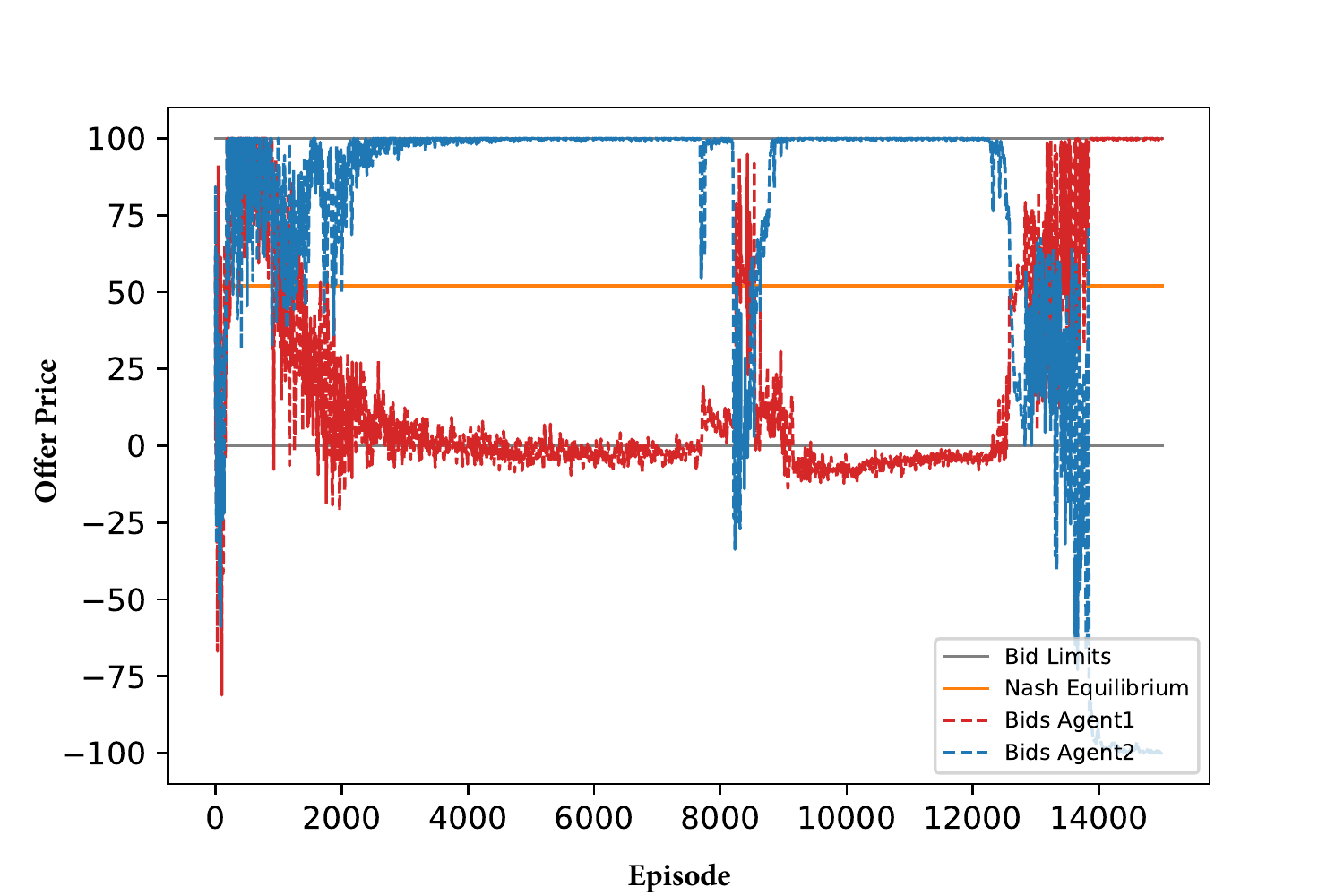}                 
                                                                            & \includegraphics[width = 1.8in]{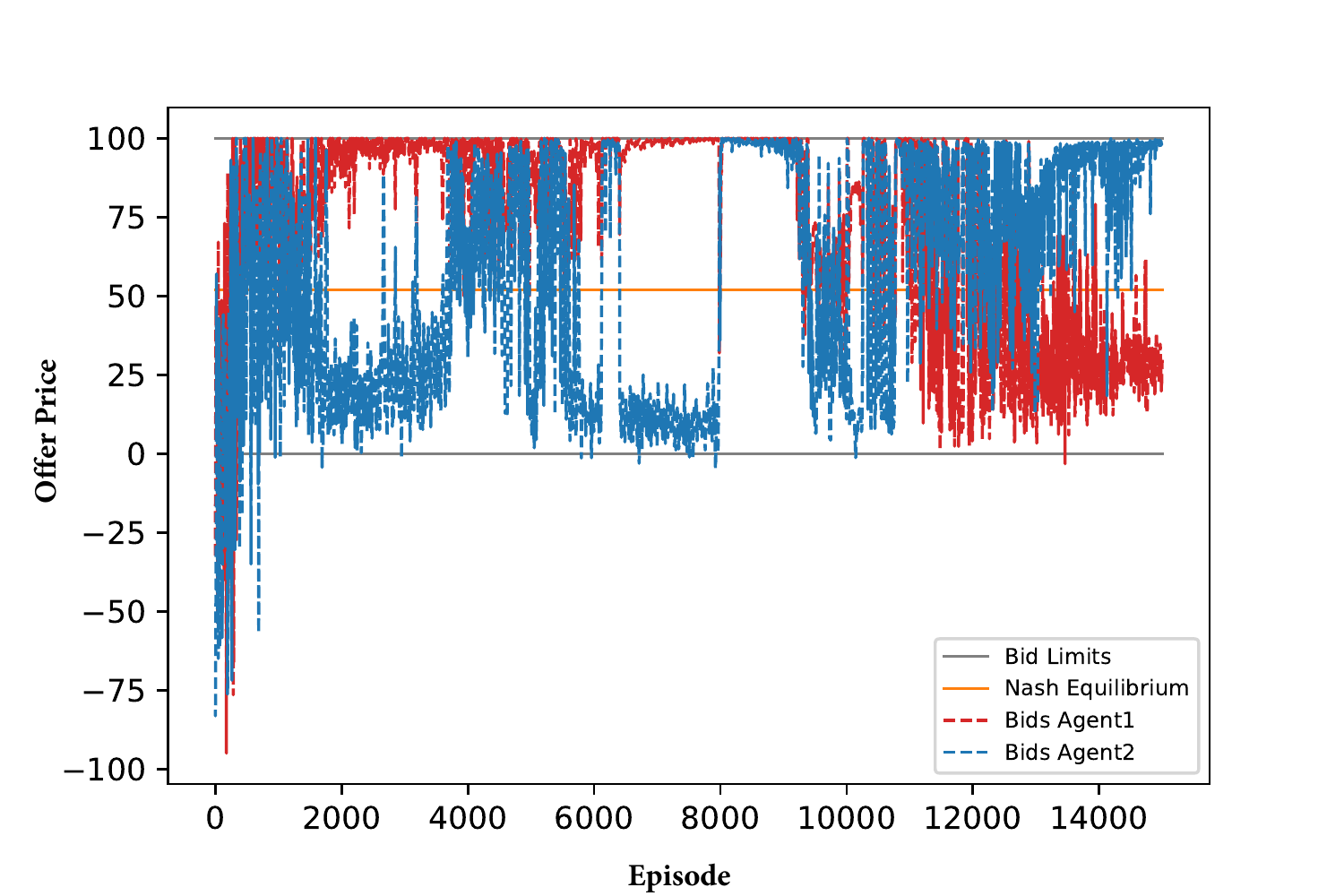}                 \\
\multicolumn{1}{|l|}{}                                                       & \includegraphics[width = 1.8in]{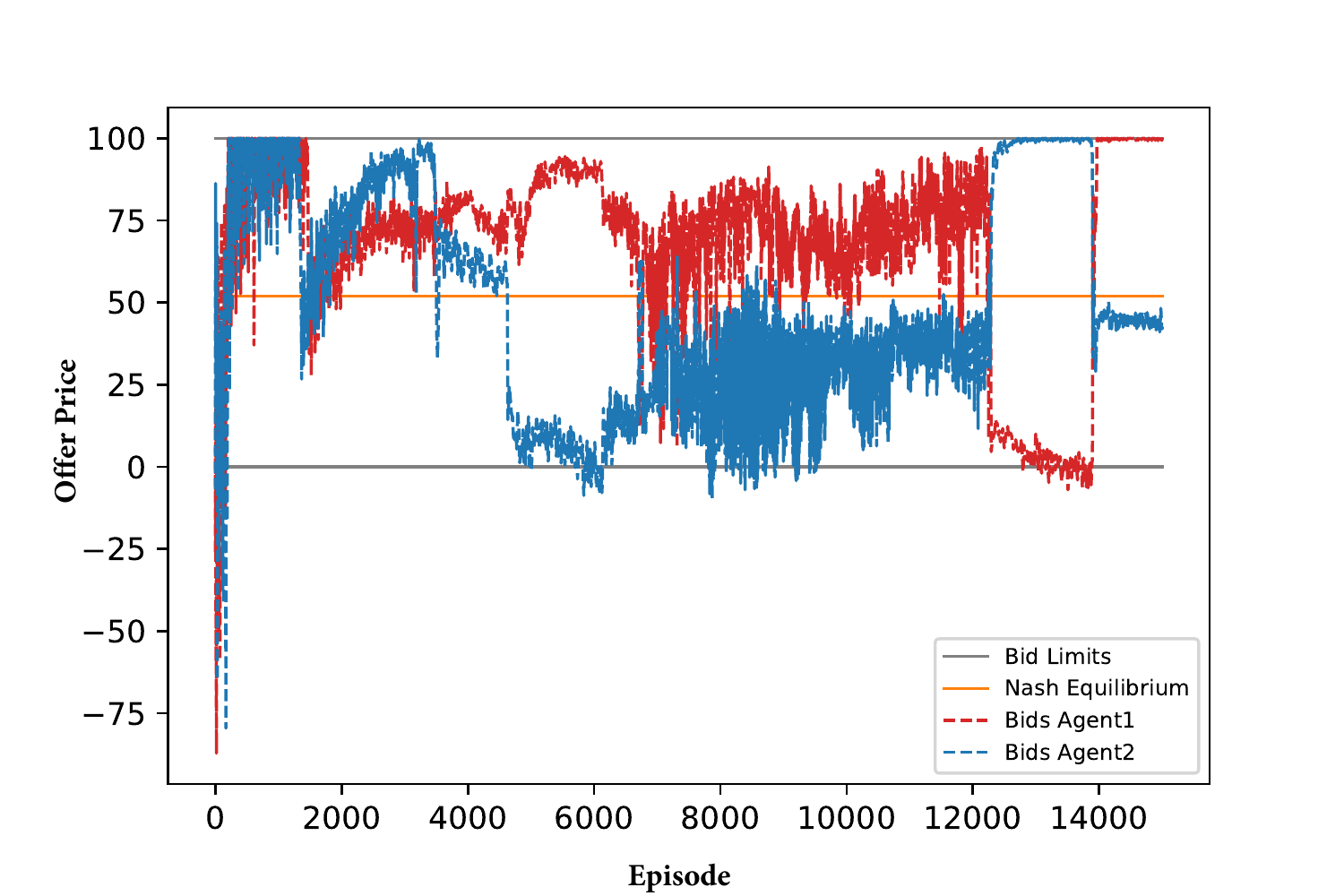}             
                                                                            & \includegraphics[width = 1.8in]{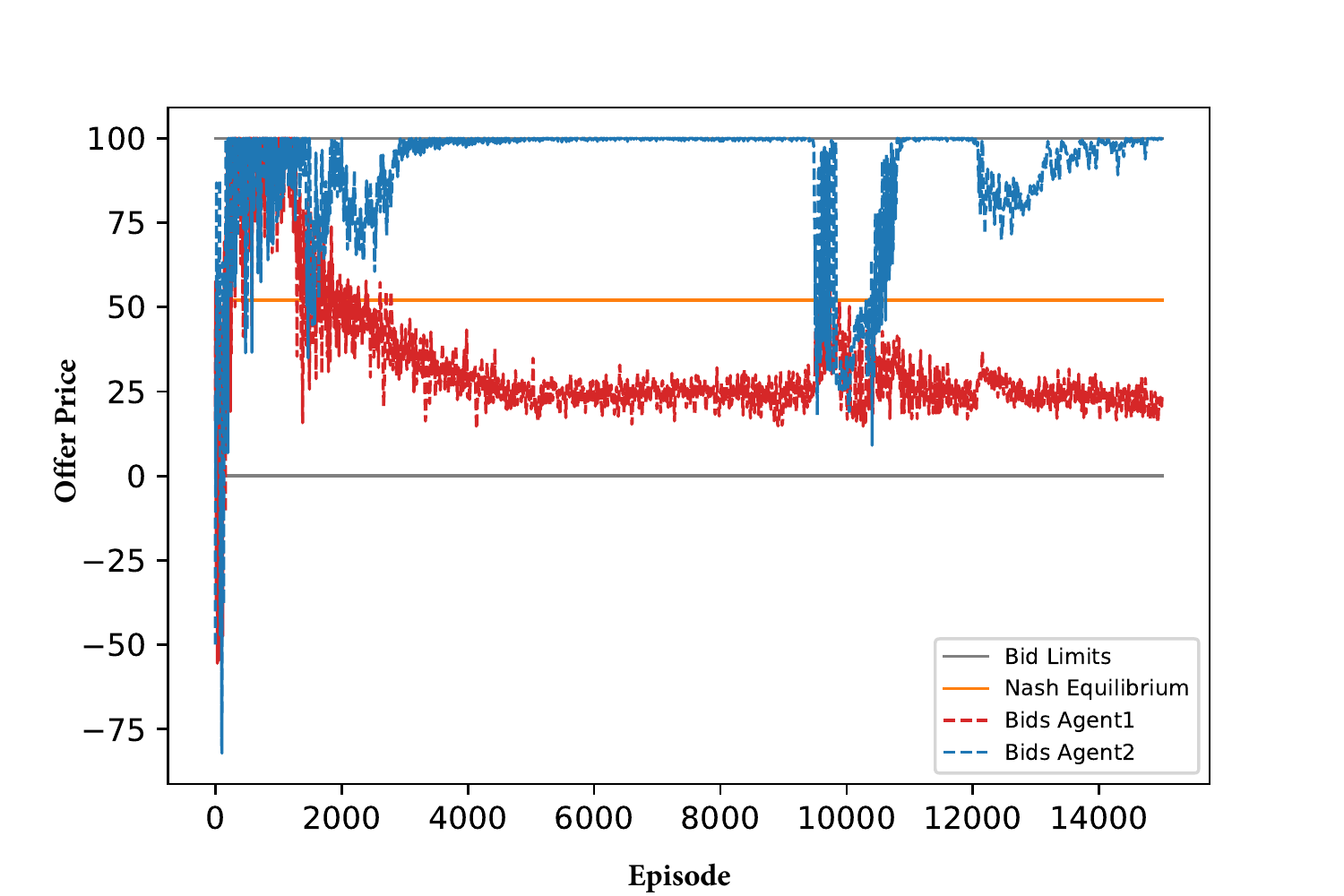}                
                                                                            & \includegraphics[width = 1.8in]{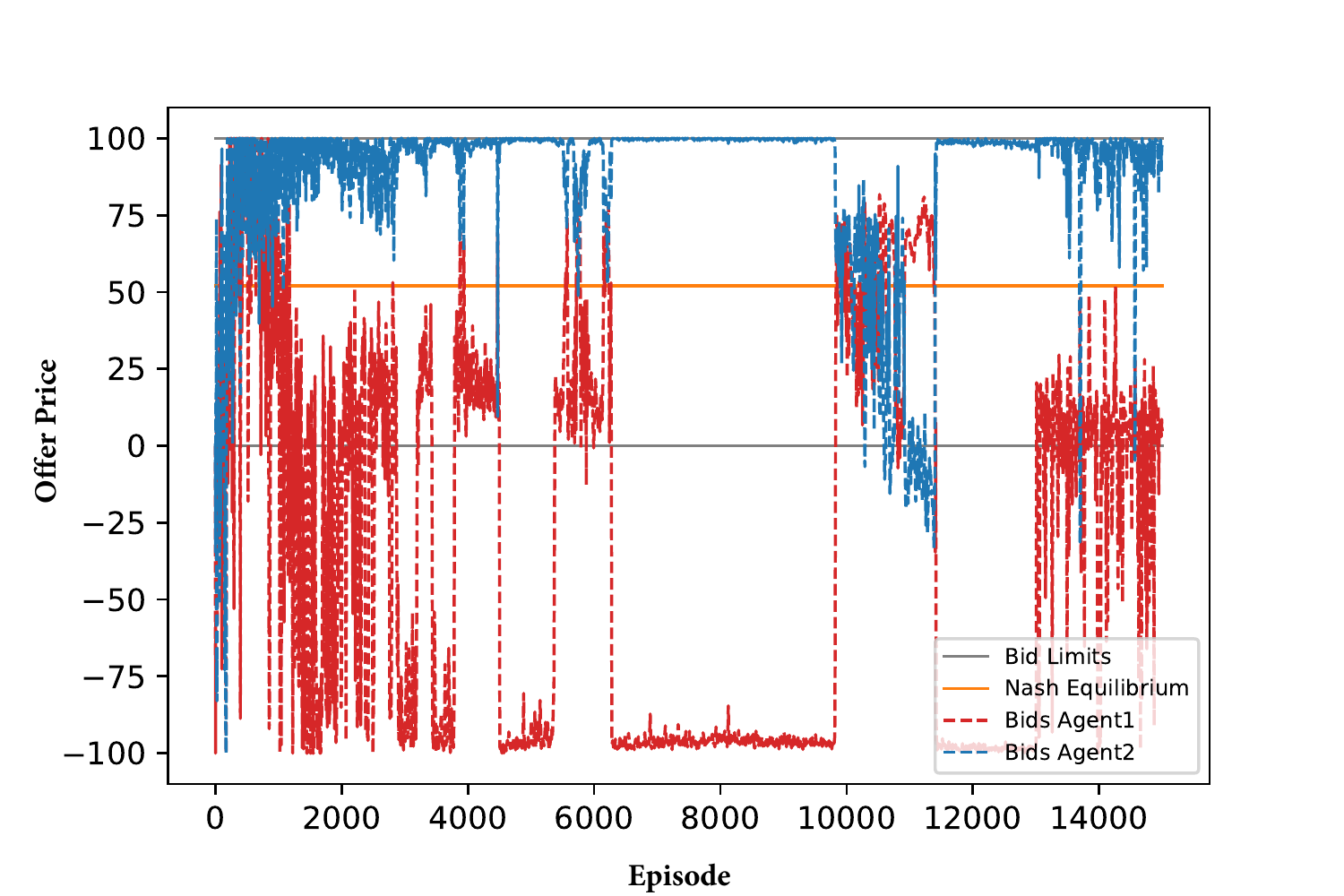}                \\ \hline
\end{tabular}
\floatfoot{\footnotesize \textit{Notes:}
Exemplary selection of 2 out of 100 DDPG runs per normalization methods and memory. Top-row plots runs without state-space memory, bottom-row plots runs with state-space memory. We believe, that in the runs with memory stronger competition is apparent. This assertion is strengthened by  Figure~\ref{Fig:Competitiveness}'s quantitative evaluation of the average number of undercuts.}
\caption{Qualitative depiction of normalization schemes and memory influence on bidding behavior}
\label{Fig:QualitativeGrid}
\end{table}

Overall, to our surprise the choice of normalization seems to strongly influence behavioral patterns of the agents.
We believe that one of the most striking facts is the influence of memory. In Table~\ref{Fig:QualitativeGrid}, we see significantly more competition between both participants, whenever we include memory in the state space. 
The memory less runs tend to settle quickly into a Nash equilibrium and exhibit an overall more static behavior.
This seems to come along with a tendency to converge easily.

This observation is confirmed by the statistics from Figures~\ref{Fig:Competitiveness}. In all cases, we see considerable more variance in the means of the cases with memory than without, although the range of offers within one standard deviation is only slightly enlarged. This illustrates an overall more unsteady behavior in the non-empty state spaces. 
The increased variance of means with memory seems to be consistent throughout the different normalization schemes.

Even though we have seen in Section~\ref{subsec:LR_parameters} that memory leads to larger sensitivity of convergence towards learning rate choice, we believe that this is not necessarily a sign of bad learning. Indeed, we are inclined to conjecture that the memory allows to know the opponents behavior better and thus engage in fiercer competition.
From the perspective of market simulation this may be a more realistic behavior or at least it may be seen as parametrizing some sort of risk-seeking. Hence, we investigate the competitveness of the two alorithms as a distinct feature of interest.

Initially, it is not clear how to measure competitiveness. We point out that a high variance is not sufficient for ensuring competitveness. For instance, consider the batch-normalized no memory examples depicted in Table~\ref{Fig:QualitativeGrid}. There, both low bidding players vary there bids considerable, but their price offers never reach the proximity of the high bidding player after episode 2000. This is an example of relatively high variance with no competiteveness. 
 
Hence, we propose the following measure of competitiveness. If one player is the high bidding player in one turn and the low bidding player in the next turn we count this as a switch. We count the number of switches per episode and divide by the turns per episode, yielding the average switches per episode. This relative measure is normalized to range in $[0,1]$ with 1 corresponding to a switch in every turn and 0 to never switching. We believe, that this measure clearly captures the notion of competitive behavior, due to being high precisely when opponents undercut each other frequently. 
 
One caveat of the measure may be that it not explicitly considers the magnitude of undercuts, but we still belief it to be informative.

Finally, we smooth the measure by applying a rolling average over a window of 100 episodes. Smoothing is not conceptual, but serves a better visual representation, since the significant variance of the competitiveness measure between episodes leads to indistinguishable graphs.
 
We depict the development of the algorithms competitiveness measure in Figure \ref{Fig:Competitiveness}. Each colour represents a different normalization method. Panel (a) depicts 100 runs without state-space memory, while Panel (b) depicts 100 runs with state-space memory. 

We point out that DDPG is equipped with two possibly interacting types of memory: the replay buffer and a possible state-space memory. Hence, results of Panel (a) are not to be understood as memoryless, however they do not include the actions of last round in the state-space, but only rely on the statistical representation of the history within the replay buffer of size 50,000. Nontheless, we see a significant effect of state-space memory on the algorithms comeptitiveness.
 
First, it has to be noted that we find unormalized runs to be the most competitive ones, while normalization tends to decrease competitiveness overall. Hence, normalization methods can severky affect the resulting outcome.
 
Generally, state-space memory does not strongly impact the development of unnormalized or layer normalized runs. This stand in stark contrast to batch normalized runs, where competitevness changes completely. While Panel (a) sees batch normalized runs by far the most uncompetitive, in Panel (b) the memory alleviates the normalization effects leaving batch and unnormalized runs almost as competitive. Generally, we believe that a normalization method is preferable if its behavior is closer to the unnormalized case introducing another quality to consider besides pure convergence rates.

Contrasting these results with the convergence rates of Figures \ref{Subfig:3norm_nomemory} and \ref{Subfig:3norm_memory}, we see that while layer normalization leads to perfect convergence and admits consistent behavior regardless of the state-space, it comes at the price of a significant decrease in competitiveness and thus essentially makes a behavioral assumption.
This is not necessarily a disadvantage, but it illustrates that the choice of the normalization method needs to be in fact a conscious modelling choice.

The decrease in convergence between layer and unnormalized runs is marginal, at least in the low feature regime that we studied. Indeed, it may be worthwhile to forfeit normalization against common machine learning practice, if competitive behavior is seen as desirable. However, the drop in convergence of unnormalized DDPG runs is significant already in the only slightly enlarged state space of Figure \ref{Subfig:3norm_memory}. It is therefore quite possible, that unnormalized DDPG is not robust in complex state spaces.

Another interesting alternative, is the use of batch normalization methods. Although, this normalization method compeletely fails both in convergence and competitiveness in Panel (a), it seems to be a sensible choice in Panel (b). If accompanied with state-space memory batch normalization is similarly in competitiveness to unnormalized runs and convergence. Moreover, it is the only method that increased its performance in the larger state-space, while all others performed worse. This may indicate a superior scaling of batch normalization to more complex state-action spaces. 
Considering, that batch normalization methods stem from computer game learning environments, where each pixel is regarded as individual state it seems at least plausible that this method truly shines only in large state spaces.
Nonetheless, this can not be fully concluded from the result at hand, but
certainly merits deeper investigation to check this conjecture. 

\begin{figure}[htb]
     \centering
     \begin{subfigure}[b]{0.49\textwidth}
         \centering
         \includegraphics[width=\textwidth]
         {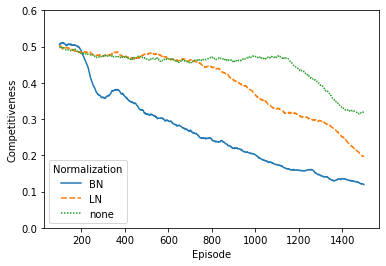}
         \caption{No State-Space Memory}
     \end{subfigure}
          \begin{subfigure}[b]{0.49\textwidth}
         \centering
     \includegraphics[width=\textwidth]         {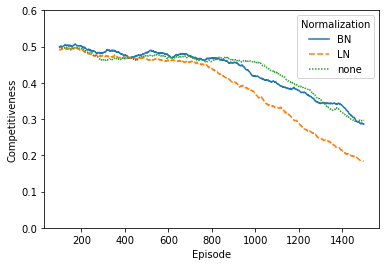}
         \caption{With State-Space Memory}
     \end{subfigure}
     \floatfoot{\footnotesize \textit{Notes:}
     To measure competitiveness, we depict the average rate of undercuts per episode from 100 DDPG runs smoothed by a rolling average over 100 episodes.
     Each plot shows 3 normalization methods: unnormalized (green), layer normalized (yellow) and batch normalized (blue). Panel (a) represents runs without explicit state-space Memory, while Panel (b) allows state-space Memory.}
    \caption{Competitiveness influenced by Normalization and Memory} 
    \label{Fig:Competitiveness}
\end{figure}

\subsection{Replay Buffer \& the limiting case of the ``ultra-competitive'' Benchmark Scenario}
\label{subsec:D=CAP}

\begin{figure}[hp]
\begin{subfigure}{\textwidth}
  \centering
  \includegraphics[width=0.7\linewidth,height=0.28\textheight]{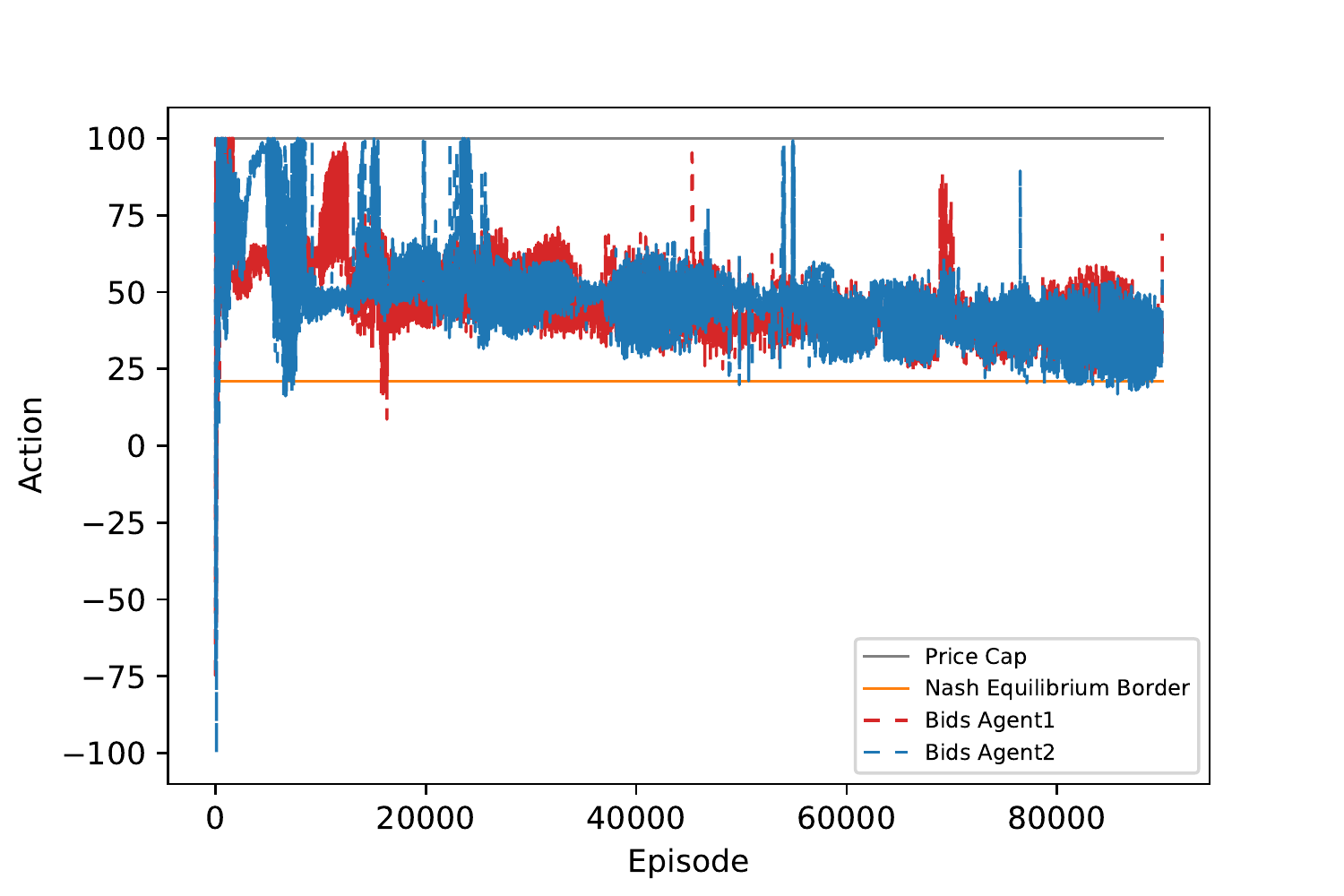}
  \caption{50,000 Round Replaybuffer}
  \label{fig:50000membuf}
\end{subfigure}
\begin{subfigure}{\textwidth}
  \centering
  \includegraphics[width=0.7\linewidth,height=0.28\textheight]{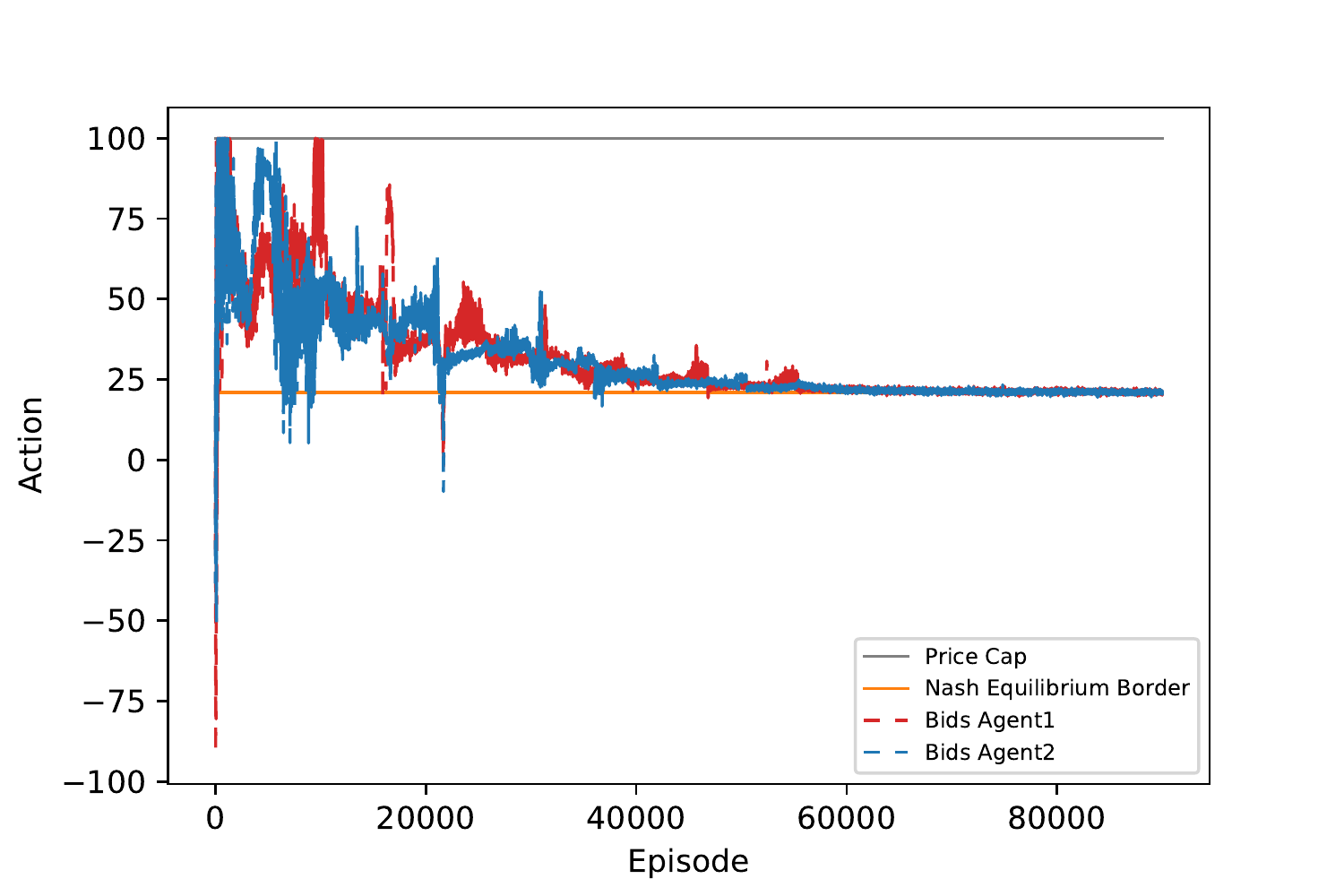}
  \caption{5,000 Round Replaybuffer}\label{fig:5000membuf}
\end{subfigure}
\begin{subfigure}{\textwidth}%
  \centering
  \includegraphics[width=0.7\linewidth,height=0.28\textheight]{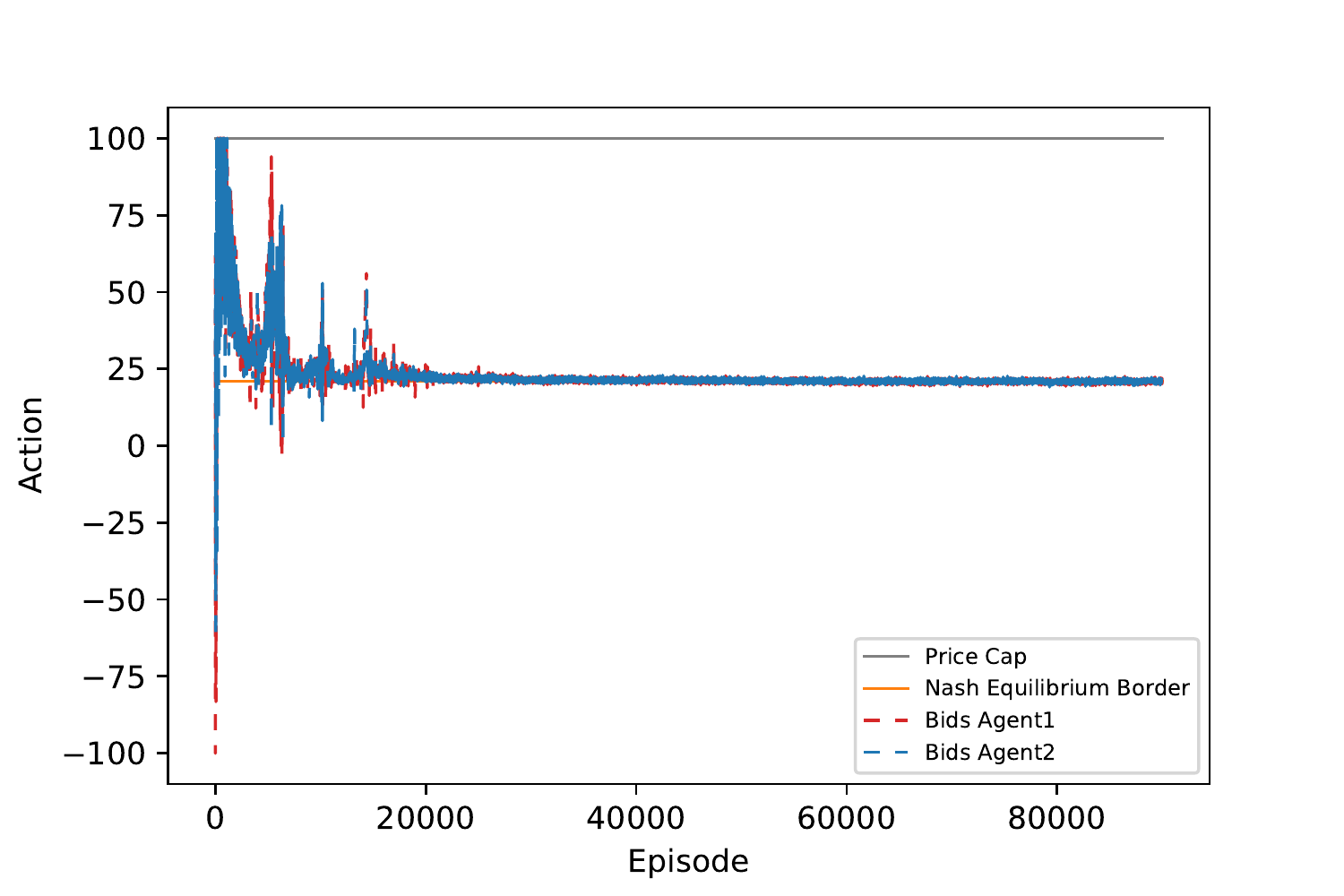}
  \caption{500 Round Replaybuffer}
\label{fig:500membuf}
\end{subfigure}
\floatfoot{\footnotesize \textit{Notes:} We showcase the influence of Memorybuffer size in an ``ultra-competitive'' Scenario, where $D \leq \overline{q}$. 
Each Panel depicts the quickest converging DDPG run out of 10 identical runs. 
Here, the convergence times are from top to bottom (a): above 100,000 Episodes, (b): around 50,000 Episodes and (c): around 20,000 Episodes, hence covergence times are \textit{decreasing with smaller replay buffer sizes}. }
\caption{DDPG-Convergence times of unconstrained Bertrand games depending on Replaybuffer size}
\label{fig:memorybuffertest}
\end{figure}

In our standard Benchmark scenario with demand $\overline{q} < D < 2\overline{q}$ converge could be reliably achieved with adequate choices of learning rates and normalization methods. In this case, a replay buffer that was comparably large with up to 50,000 turns memorized was sufficient to attain convergence.
It has to be said, that in the standard scenario 20,000 episodes was used as a standard runtime with algorithms frequently converging already around 5,000 episodes.

Classically, Bertrand duopolies are studied unconstrained and symmetric. This situation is well known to exhibit a unique equilibrium, where both agents bid their marginal costs. We find DDPG to generally converge to this solution but found that the configuration of the replay buffer may throttle the speed of convergence substantially. Even extremely long runtimes of up to 100,000 episodes did not converge fully until we reduced the size of the replay buffer. Indeed, in this situation the size of the replay buffer seems absolutely critical. Figure \ref{fig:memorybuffertest} Panel (a)--(c) depict the effect of memory buffer size on DDPG's convergence time.

Another, parameter that has been proven critical is the noises decay rate. This is less surprising since highly competitive equilibria require sufficient exploration to converge. The ``super-competitive'' scenario is such an highly competitive equilibrium. Consequently, the effect of the decay rate is clearly visible on the convergence of unconstrained Bertrand games.

\subsubsection{Replay Buffer Size}

We tested replay buffer sizes of 50,000 (ca 400 Batches), 5,000 (ca 40 Batches), 500 (ca 4 Batches), 256 (2 Batches) and 129 (1 Batch) rounds each.

Neither very high (50,000) nor very low choices (129,256) led to good results.\footnote{The choice of low memory buffer sizes are determined as batch size plus one and batch size times two.}
Too large memory buffers led to both agents only marginally undercutting each other with a steady but extremely slow decrease in offers over time. Smaller choices of memory buffer corresponded directly to fast decreases in offers towards the marginal cost equilibrium. However, too small memory buffers (129,256) quickly approached marginal costs but did not necessarily stay within equilibrium. Overall, buffer sizes of 500 converged quickly but yet remained stable. We want to stress that this trend held true both \textit{with and without the inclusion of past actions into the state-space}.

We explain this phenomena by the fact that large buffer sizes slow learning. For instance, in the $D\leq \overline{q}$ case, the equilibrium is only found through long periods of competition. However, initially both algorithms can succeed with relatively large offer prices. Big buffers tend to keep up old experiences for a relatively long time and result in the algorithm decreasing its offer prices slower than with small buffers, as they remember initial successes longer. In contrast, if the memory buffer is not big enough to retain the information that deviations from equilibrium are punished, the algorithm will be instable as it begins to explore again after convergence to equilibrium.
This indicates that replay buffers are indeed an integral part of DDPG and can not be simply left out. 
Nonetheless, the buffer size should be carefully limited. If convergence is possible, smaller memories are better.

\subsubsection{Noise \& Decay Rate}

As its name suggests, deep deterministic policy gradient
(DDPG) algorithms employ neural networks that deterministicly output the same values given the same inputs.
Outputs only change after training on unseen data.
However, novel data would never be explored by a determinstic neural netowrk.
Noise introduces randomness into the algorithm, in order to generate novel behavior.
Hence, the noise directly controls DDPG's exploration. 
In turn, the decay rate controls the fading out of the initially high noise over time that ensures an eventual switch from intial exploration to exploitation.
Hence, the decay rate is the most direct correspondence to $Q$-learnings exploitation-exploration parameter.

We depict the influence of the decay rate on convergence behavior in Figure \ref{fig:decayrate}. The figure shows relatively common behavior similar to many reinforcement algorithms.
This is not surprising since the interplay between exploration and exploitation is a shared characteristic of all RL algorithms.
Too fast decay rates, such as 0.9, lead to trivial behavior, while slow  decay rates such as 0.99999 leasto complete random behavior. In the intermediate regime, convergence is possible with the most-balanced choice of 0.999 leading to the best convergence time.

\begin{figure}[hp]
\begin{subfigure}{\textwidth}
  \centering
  \includegraphics[width=0.7\linewidth,height=0.28\textheight]{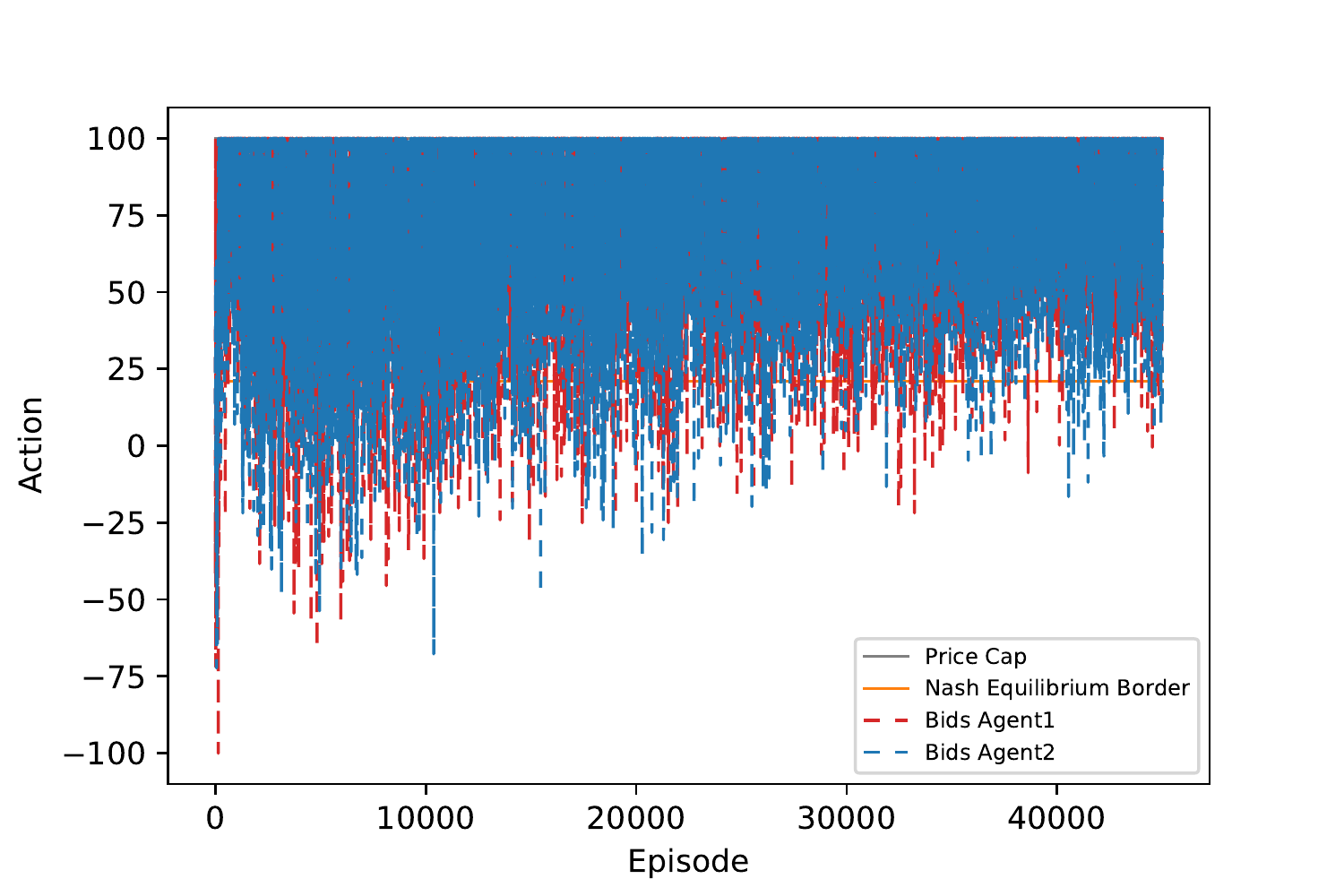}
  \caption{0.9999 Decay Rate}
\end{subfigure}
\begin{subfigure}{\textwidth}
  \centering
  \includegraphics[width=0.7\linewidth,height=0.28\textheight]{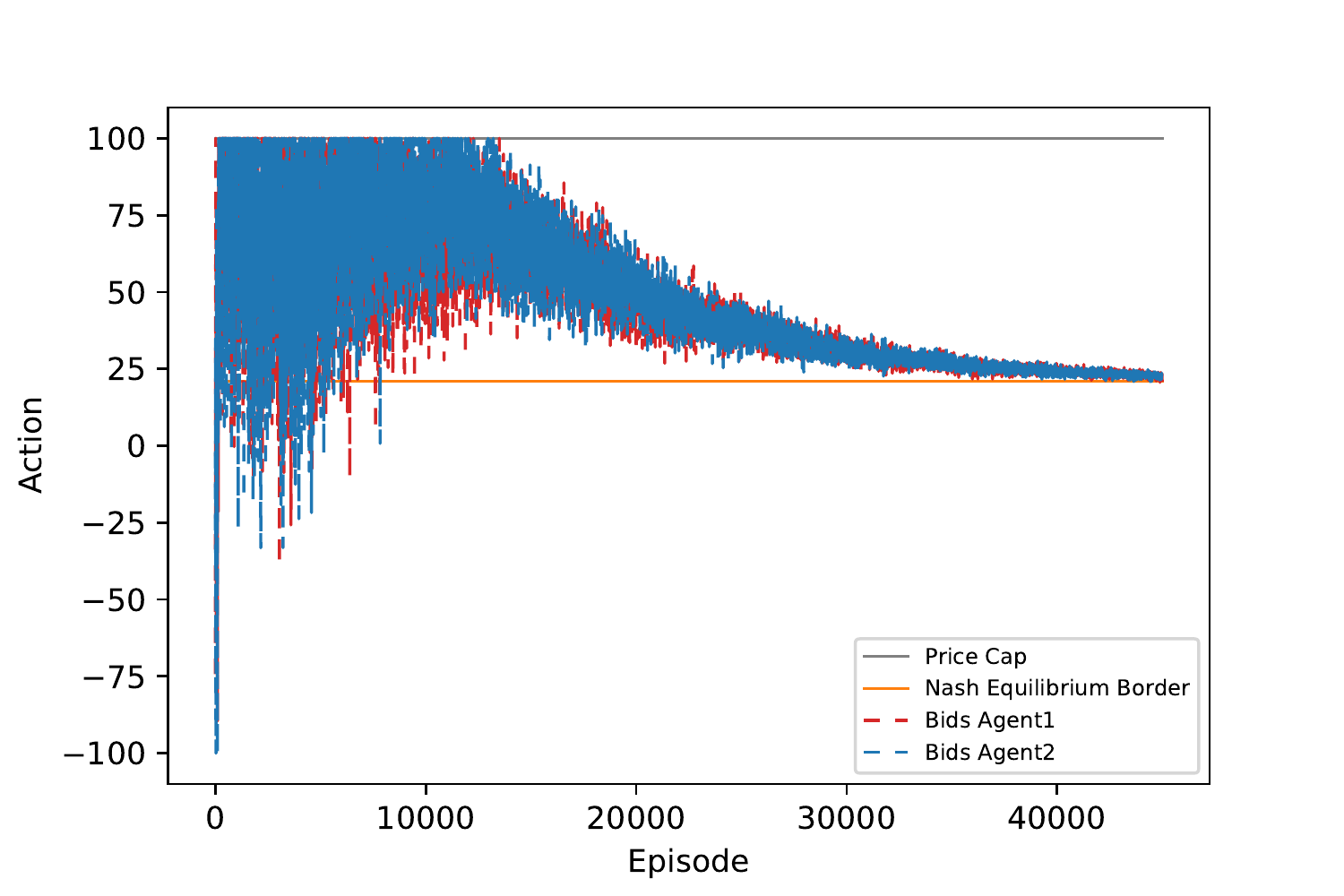}
  \caption{0.999 Decay Rate}
\end{subfigure}
\begin{subfigure}{\textwidth}
  \centering
  \includegraphics[width=0.7\linewidth,height=0.28\textheight]{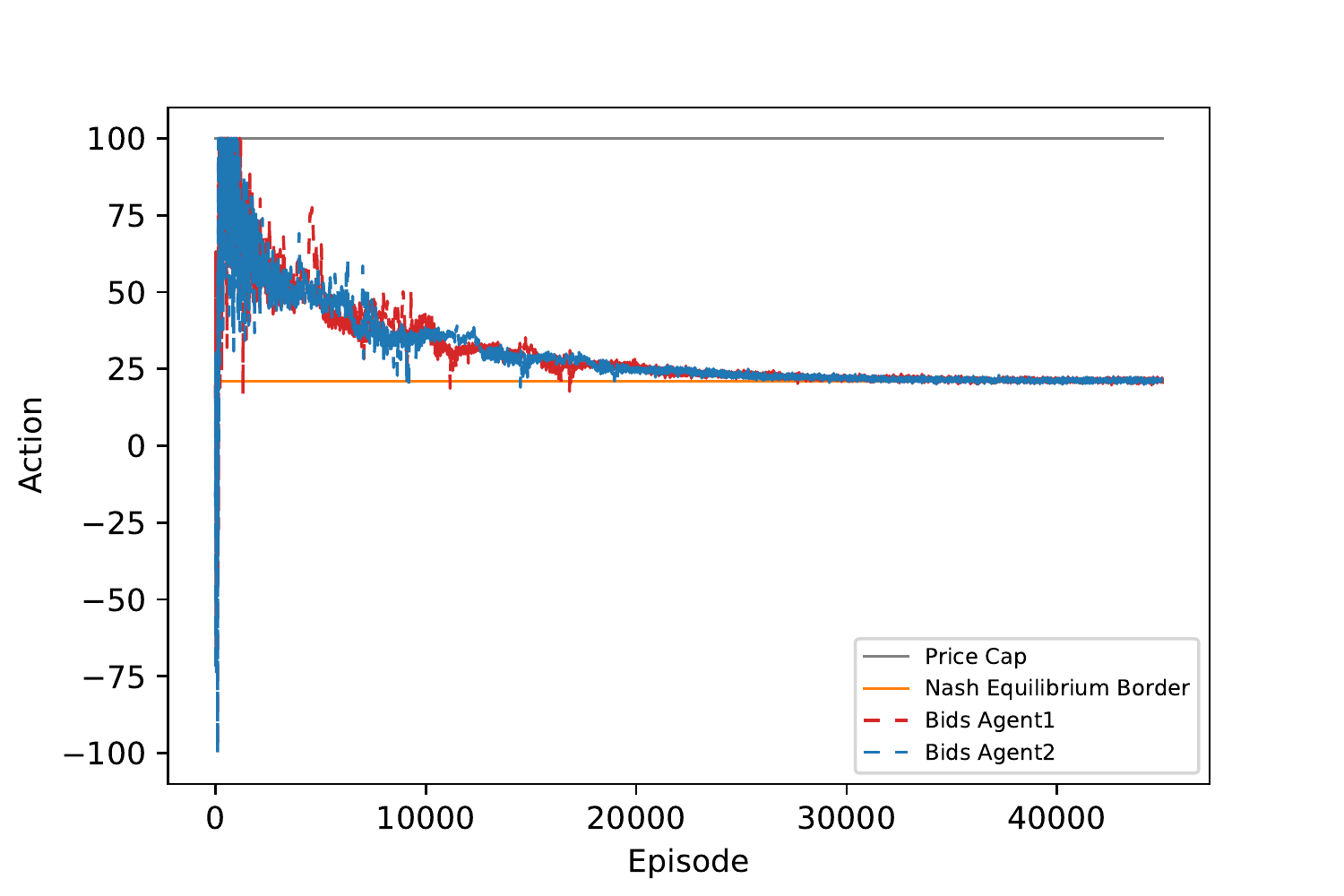}
  \caption{0.99 Decay Rate}
\end{subfigure}
  \floatfoot{\footnotesize The figure depicts identical DDPG runs with varying decay rates. The decay speed increases from top to bottom. Note that decay rates are multiplicative factors, hence larger number correspond to slow decay.. The depicted rates are
  (a) 0.9999, (b) 0.999 and (c) 0.99. Even faster decay rates do have a contraproductive effect on convergence times.}
\caption{DDPG-Convergence times of unconstrained Bertrand games depending on decay rate size}
\label{fig:decayrate}
\end{figure}

\section{Conclusion}
\label{Sec:conclusion}

Our main contribution is the \textit{adaption of the deep deterministic policy gradient (DDPG)} algorithm to a novel problem domain: \textit{strategic interaction in uniform price auctions}.\footnote{If players do not have capacity constraints these setting translates to price competition \`a la Bertrand.} 
The problem has been addressed by classical reinforcement learning methods, however we are to the best of our knowledge the first to apply modern deep reinforcement learning techniques to this particular problem domain.
We believe and argue that DDPG is a significant step beyond the usage of the frequently used $Q$-learning. 
The methodological novelty lies in the possibility of allowing fully continuous state and action spaces (as compared to $Q$-learning) and no need for any assumptions of a functional form parametrizing the strategy space of the agents (as compared to policy gradient learning).
We have argued that despite frequently being labelled ``model-free'' $Q$-learning's state-action space assumptions indeed require significant model assumptions to be justified.
Essentially, DDPG's continuity properties are another step-forward towards being truly model free. Nonetheless, we have pointed out the caveat, that certain parametric choices such as memory models and normalization methods do impact agent behaviors.

Furthermore, the problem is naturally a multi-agent problem and applying the initially single agent DDPG to multi-agent settings is in itself a timely endeavour, although there have been already several works pursuing this  new line of research. In summary, we believe to have demonstrated that modern learning algorithms like DDPG essentially remove the necessity to discretize strategic bidding state-action spaces at least in low feature spaces. Most likely, this result is generalizable to scenarios with numerous features, however this needs to be ascertained by future work.

We accompany our DDPG implementation with a benchmark framework based on its ability to consistently converge to Nash equilibria.
This is only possible through our employment of a scenario (see Section \ref{Sec:BenchmarkingScenario}) where we can analytically derive Nash equilibria. In our case, we consider a duopoly engaged in a uniform price auction of a capacity constrained homogeneous good. We are able to completely characterize the auctions equilibrium strategies. We have performed a number of DDPG market-simulations while varying the key learning parameters \textit{actor learning rater, critic learning rate, state space memory} and \textit{normalization technique}. Each parametrization has been run for 100 times and evaluated with respect to the percentage of strategies that converged to an equilibrium after 15,000 episodes of learning. This allows us to asses the impact of the aforementioned parameters on the learnability of the equilibria.

Despite considerable variance in the results between differing algorithm parametrizations there is a common trend: 
Well parametrized DDPG-simulation play equilibrium strategies almost always (more than 95\% of the runs).
Ill parametrized runs may not converge at all.
Nonetheless, our first key finding is: it is possible to reliable find equilibrium strategies with a properly tuned DDPG-algorithm.
Commonly, deep learning algorithms successes depend crucially on the correct parametrization. 
Hence, we give specific advise how to choose these parameters within Section~\ref{subsec: parameter variation}. 

Furthermore, we discuss the effects 
of \textit{explicit state-space memory}, \textit{replay buffer memory} and \textit{normalization method} 
encountered during our analysis. All showed significant impacts on the algorithms behavior. 
Furthermore, layer normalization clearly performed best in our benchmark if measured solely in terms of convergence reliability.

Including an \textit{explicit memory} of the last rounds actions as state information had counter-intutive results. Similarly, \textit{larger replay buffers} did not necessarily perform better. To the contrary, larger replay buffers led to slower convergence times. One might have expected memory to ``help'' the algorithm and thus converge more easily, but we found the opposite effect. Memoryless algorithms converged in a wider range of learning rates. Nonetheless, both memoryless and algorithms with memory reached convergence rates above 90\%, albeit the memoryless ones did so more robustly.

This certainly merits further investigation, however we conjecture two possible reasons for this behavior. First, the large state space of the memory might counteract any informational benefits. This would indicate that one should include state-space information only if one is certain that it is vital information.
Second, the algorithms with memory appear to compete more proactively. Possibly, awareness of the opponent is necessary to place competitive bids. It may be the case that competition seeking agents explore the strategy space more thoroughly and thus require longer equilibration times. In our opinion, this could be beneficial for market simulations or policy analysis, where one might hope to fully explore the strategy space.
Overall, this leads us to the assessment that with respect to algorithmic memory less seems to be more.

Furthermore, we have to stress the impact of \textit{normalization schemes} on the final behavior. Generally, unnormalized runs did perform well too, however the best convergence rates were encountered by the use of so-called layer normalization. Layer normalization tended to produce well-equilibrated, but also relatively static and non-competitive behavioral patterns. This stood in contrast to batch normalization.
Batch normalization showed a highly differentiated behavior.
While memoryless batch normalization led to the by far worst convergence rates, batch normalization with memory performed satisfyingly albeit worse than layer normalization with memory. However, it was the only method that profited from the inclusion of memory. this may be an indication that batch normalization is better suited to handle more complex state-space and we will pursue this matter in the future. Moreover, it showed an extremely pronounced tendency to compete, experiment, and explore the state-space. This may or may-not be desirable, but we certainly see significant behavioral assumptions entering through seemingly inconspicuous normalization methods. Hence, we advise to take this into account. 

In summary, we believe to have demonstrated that Nash equilbria in multi-agent uniform price auctions can be found by DDPG simulations. This makes it a promising tool to analyze auctions or to derive informative counter-factuals even in more general settings.

\bibliography{lit}

\end{document}